\documentclass[12pt]{article}
\usepackage{amsmath}
\usepackage{amsfonts}
\usepackage{makeidx}
\usepackage{amssymb}
\usepackage{times}

\newtheorem{satz}{Theorem}[section]
\newtheorem{definition}[satz]{Definition}

\newtheorem{koro}[satz]{Corollary}
\newtheorem{bemerkung}[satz]{Remark}

\newtheorem{notation}[satz]{Notation}
\newenvironment{proof}{\par\noindent {\it Proof:} \hspace{7pt}}{\hfill\hbox{\vrule width 7pt depth 0pt height 7pt}
\par\vspace{10pt}}

\begin{document}

\title{Microscopic Conductivity of Lattice Fermions\ at Equilibrium -- Part
II: Interacting Particles}
\author{Jean-Bernard Bru \and Walter de Siqueira Pedra}
\date{\today }
\maketitle

\begin{abstract}
We apply Lieb--Robinson bounds for multi--com%
\-%
mutators we recently derived \cite{brupedraLR} to study the (possibly
non--linear) response of interacting fermions at thermal equilibrium to
perturbations of the external electromagnetic field. This analysis leads to
an extension of the results for quasi--free fermions of \cite{OhmI,OhmII} to
fermion systems on the lattice with short--range interactions. More
precisely, we investigate entropy production and charge transport properties
of non--autonomous $C^{\ast }$--dyna%
\-%
mical systems associated with interacting lattice fermions within bounded
static potentials and in presence of an electric field that is time-- and
space--dependent. We verify the 1st law of thermodynamics for the heat
production of the system under consideration. In linear response theory, the
latter is related with Ohm and Joule's laws. These laws are proven here to
hold at the microscopic scale, uniformly with respect to the size of the
(microscopic) region where the electric field is applied. An important
outcome is the extension of the notion of conductivity measures to
interacting fermions.
\end{abstract}



\tableofcontents%

\section{Introduction}

The present paper belongs to a succession of works on Ohm and Joule's laws
starting with \cite{OhmI,OhmII,OhmIII,OhmIV}. These papers give a complete
and mathematically rigorous derivation (at least in the AC--regime) of the
phenomenon of linear conductivity from microscopic quantum dynamics and
first principles of thermodynamics, only. Their results lead, in particular,
to a physical picture of the microscopic origin of Ohm and Joule's laws
based on a notion of \textquotedblleft quantum viscosity\textquotedblright\
for currents, highlighting the role of so--called diamagnetic and
paramagnetic currents. However, a drawback of our previous studies \cite%
{OhmI,OhmII,OhmIII,OhmIV} is their restriction to non--interacting fermions
in disordered media. Indeed, it is believed in theoretical physics that
electric resistance of conductors should also result from interactions
between charge carriers, and not only from the presence of inhomogenities
(impurities).

Therefore, we aim to extend the results of \cite{OhmI,OhmII,OhmIII,OhmIV} to
fermion systems with interactions, i.e., to rigorously derive, at least in
the AC--regime, Ohm and Joule's laws for interacting systems. See also \cite%
{brupedrahistoire} for a historical perspective of this subject. The present
paper is a first step in this direction, extending all results of \cite%
{OhmI,OhmII} to fermion systems on the lattice with short--range
interactions and bounded static potentials:

\begin{itemize}
\item We investigate the heat production for some non--auto%
\-%
nomous $C^{\ast }$--dyna%
\-%
mical systems on the CAR $C^{\ast }$--algebra of cubic infinite lattices of
any dimension with respect to (w.r.t.) completely passive states. Such
states implement the 2nd law of thermodynamics and are identified here with
the (thermal) equilibrium states of the system. The (non--auto%
\-%
nomous) dynamics is generated by a short--range interaction between
particles, a bounded static potential, and next neighbor hoppings in
presence of an electromagnetic field that is time-- and space--dependent. In
particular, we verify the 1st law of thermodynamics for the system under
consideration.

\item We next derive Ohm and Joule's laws at the microscopic scale, an \emph{%
unexpected} \cite{Ohm-exp2} property experimentally verified in 2012 at the
atomic scale \cite{Ohm-exp}. More precisely, we show that, at weak external
electric fields, uniformly w.r.t. the size of the region where the electric
field is applied, the current density response is linear and the heat
production density is quadratic in the strength of the applied field. We
introduce the notion of conductivity measures for interacting fermions,
similar to \cite{Annale,JMP-autre,JMP-autre2} and \cite{OhmII,OhmIII,OhmIV}
in the non--interacting case. Note however that the detailed structure of
the conductivity measure is not studied, yet. It depends, of course, on the
system under consideration.
\end{itemize}

\noindent Note that in \cite[Section 3.5]{OhmII} so--called microscopic
\emph{conductivity distributions} are defined from conductivity measures.
Exactly the same construction, at least for time--reversal invariant models,
could be done here and we refrain from doing it again. The same remark can
be done for the derivation of Joule's law in its original formulation, see
\cite[Section 4.5]{OhmII} for more details.

Like in \cite{OhmI,OhmII}, all estimates we get are uniform w.r.t. the size
of the region with non--vanishing electric fields. This is possible because
we prove in \cite[Corollary 3.10]{brupedraLR} the tree--decay bounds on
multi--commutators of \cite[Corollary 4.3]{OhmI} for the case of interacting
fermions. These bounds are a pivotal ingredient in \cite%
{OhmI,OhmII,OhmIII,OhmIV} and were derived in \cite[Section 4]{OhmI} by
explicit computations using the fact that the fermion system was \emph{non}%
--interacting. For interacting systems, we use in \cite{brupedraLR} \emph{%
Lieb--Robinson bounds for multi--commutators }together with combinatorics of
maximally connected graphs (tree expansions)\emph{\ }to prove them. Note
additionally that Lieb--Robinson bounds for multi--com%
\-%
mutators also enter in a decisive way in the proof of the 1st law of
thermodynamics for the case of interacting particles, see for instance
Theorem \ref{main 1 copy(1)} (i) and Remark \ref{remark 1st law of
Thermodynamics}.

As discussed in \cite{brupedraLR}, this method requires short--range
interactions. The short--range property is defined via the finiteness of a
convenient norm for interactions: This norm ensures that the interaction is
sufficiently weak at large distances. The set of such short--range
interactions form a whole Banach space $\mathcal{W}$ w.r.t. this norm. The
space $\mathcal{W}$ includes density--density interactions resulting from
the second quantization of two--body interactions defined via a real--valued
and summable (in a convenient sense) function $v\left( r\right) :[0,\infty
)\rightarrow \mathbb{R}$. Considering fermions with spins $\uparrow $ or $%
\downarrow $, our setting includes for instance the celebrated Hubbard model
(and any other system with finite range interactions) or models with
Yukawa--type potentials, but the Coulomb potential is excluded because it is
not summable in space.

Our main assertions are Theorems \ref{main 1 copy(1)}--\ref{lemma sigma pos
type} and \ref{lemma sigma pos type copy(4)}--\ref{Local Ohm's law thm
copy(2)}. This paper is organized as follows:

\begin{itemize}
\item Section \ref{sect 2.1 copy(1)} precisely formulates our setting. We
define in particular a Banach space of short--range interactions.

\item Section \ref{Heat Production and Current Linear Response} extends
\emph{all} results of \cite{OhmI,OhmII} to fermion systems within bounded
static potentials and with interactions decaying fast enough in space.
Observe that, in contrast to \cite{OhmII}, we generally do not have
time--reversal invariance. This yields some (straightforward) adaptations of
formulations of results and arguments.
\end{itemize}

\begin{notation}
\label{remark constant}\mbox{
}\newline
To simplify notation, we denote by $D$ any generic positive and finite
constant. These constants do not need to be the same from one statement to
another. A norm on a generic vector space $\mathcal{X}$ is denoted by $\Vert
\cdot \Vert _{\mathcal{X}}$ and the identity map of $\mathcal{X}$ by $%
\mathbf{1}_{\mathcal{X}}$.
\end{notation}

\section{$C^{\ast }$--Dynamical Systems for Interacting Fermions\label{sect
2.1 copy(1)}\label{Section main results}}

\subsection{Free Fermions within Electromagnetic Fields\label{sect 2.1}}

The host material for conducting fermions is assumed to be a cubic crystal
represented by the $d$--dimensional cubic lattice $\mathfrak{L}:=\mathbb{Z}%
^{d}$ ($d\in \mathbb{N}$).

Disorder in the crystal will be modeled by a random variable with
distribution $\mathfrak{a}_{\Omega }$ taking values in the measurable space $%
(\Omega ,\mathfrak{A}_{\Omega })$, i.e., $(\Omega ,\mathfrak{A}_{\Omega },%
\mathfrak{a}_{\Omega })$ is a probability space. Let $\mathfrak{b}$ be the
set of non--oriented bonds of the cubic lattice $\mathfrak{L}$:%
\begin{equation}
\mathfrak{b}:=\{\{x,x^{\prime }\}\subset \mathfrak{L}\text{ }:\text{ }%
|x-x^{\prime }|=1\}\text{ }.  \label{n-o bonds}
\end{equation}%
Then, $\Omega :=[-1,1]^{\mathfrak{L}}\times \mathbb{D}^{\mathfrak{b}}$ with $%
\mathbb{D}:=\{z\in \mathbb{C}$ $:$ $\left\vert z\right\vert \leq 1\}$. I.e.,
any element of $\Omega $ is a pair $\omega =\left( \omega _{1},\omega
_{2}\right) \in \Omega $, where $\omega _{1}$ is a function of lattice sites
with values in $[-1,1]$ and $\omega _{2}$ is a function of bonds with values
in the complex closed unit disc $\mathbb{D}$. The first component $\omega
_{1}$ is related to the random external static potential and the second
component to random hopping amplitudes of fermions. In a subsequent paper we
will explicitly fix the probability space $(\Omega ,\mathfrak{A}_{\Omega },%
\mathfrak{a}_{\Omega })$. In the present work, however, all studies are
performed for any \emph{fixed} realization $\omega \in \Omega $ and the
specific probability space is unimportant. All results here are uniform
w.r.t. the choice of $\omega \in \Omega $.

For any $\omega =\left( \omega _{1},\omega _{2}\right) \in \Omega $, $%
V_{\omega }\in \mathcal{B}(\ell ^{2}(\mathfrak{L}))$ is by definition the
self--adjoint multiplication operator with the function $\omega _{1}:%
\mathfrak{L}\rightarrow \lbrack -1,1]$. It represents a bounded static
potential. To all $\omega \in \Omega $ and strength $\vartheta \in \mathbb{R}%
_{0}^{+}$ of hopping disorder, we also associate another self--adjoint
operator $\Delta _{\omega ,\vartheta }\in \mathcal{B}(\ell ^{2}(\mathfrak{L}%
))$ describing the hoppings of a single particle in the lattice:%
\begin{eqnarray}
\lbrack \Delta _{\omega ,\vartheta }(\psi )](x) &:=&2d\psi (x)-\sum_{j=1}^{d}\Big((1+\vartheta \overline{\omega _{2}(\{x,x-e_{j}\})})\ \psi (x-e_{j})
\notag \\
&&+\psi (x+e_{j})(1+\vartheta \omega _{2}(\{x,x+e_{j}\}))\Big)
\label{discrete laplacian}
\end{eqnarray}%
for any $x\in \mathfrak{L}$ and $\psi \in \ell ^{2}(\mathfrak{L})$, with $%
\{e_{k}\}_{k=1}^{d}$ being the canonical orthonormal basis of the Euclidian
space $\mathbb{R}^{d}$. In the case of vanishing hopping disorder $\vartheta
=0$ (up to a minus sign) $\Delta _{\omega ,\vartheta }$ is the usual $d$%
--dimensional discrete Laplacian. Since the hopping amplitudes are complex
valued ($\omega _{2}$ takes values in $\mathbb{D}$), note additionally that
random electromagnetic potentials can be implemented in our model.

Then, for any $\omega \in \Omega $ and parameters $\lambda ,\vartheta \in
\mathbb{R}_{0}^{+}$, the Hamiltonian describing one quantum particle within
a bounded static potential is the discrete Schr\"{o}dinger operator $(\Delta
_{\omega ,\vartheta }+\lambda V_{\omega })$ acting on the Hilbert space $%
\ell ^{2}(\mathfrak{L})$. The coupling constants $\lambda ,\vartheta \in
\mathbb{R}_{0}^{+}$ represent the strength of disorder of respectively the
external static potential and hopping amplitudes.

The time--dependent electromagnetic potential is defined by a compactly
supported time--depen%
\-%
dent vector potential%
\begin{equation*}
\mathbf{A}\in \mathbf{C}_{0}^{\infty }:=\underset{l\in \mathbb{R}^{+}}{%
\mathop{\displaystyle \bigcup}}C_{0}^{\infty }(\mathbb{R}\times \left[ -l,l%
\right] ^{d};({\mathbb{R}}^{d})^{\ast })\ ,
\end{equation*}%
where $({\mathbb{R}}^{d})^{\ast }$ is the set of one--forms\footnote{%
In a strict sense, one should take the dual space of the tangent spaces $T({%
\mathbb{R}}^{d})_{x}$, $x\in {\mathbb{R}}^{d}$.} on ${\mathbb{R}}^{d}$ that
take values in $\mathbb{R}$. I.e., for some $l\in \mathbb{R}^{+}$, $\mathbf{A%
}\in C_{0}^{\infty }(\mathbb{R}\times \left[ -l,l\right] ^{d};({\mathbb{R}}%
^{d})^{\ast })$ and we use the convention $\mathbf{A}(t,x)\equiv 0$ whenever
$x\notin \lbrack -l,l]^{d}$. Since $\mathbf{A}\in \mathbf{C}_{0}^{\infty }$,
$\mathbf{A}(t,x)=0$ for all $t\leq t_{0}$, where $t_{0}\in \mathbb{R}$ is
some initial time. The smoothness of $\mathbf{A}$ is not essential in the
proofs and is only assumed for simplicity.

We use the Weyl gauge (also named temporal gauge) for the electromagnetic
field and as a consequence,%
\begin{equation}
E_{\mathbf{A}}(t,x):=-\partial _{t}\mathbf{A}(t,x)\ ,\quad t\in \mathbb{R},\
x\in \mathbb{R}^{d}\ ,  \label{V bar 0}
\end{equation}%
is the electric field associated with $\mathbf{A}$. We also define the
integrated electric field (or electric tension) along the oriented bond $%
\mathbf{x}:=(x^{(1)},x^{(2)})\in \mathfrak{L}^{2}$ at time $t\in \mathbb{R}$
by%
\begin{equation}
\mathbf{E}_{t}^{\mathbf{A}}\left( \mathbf{x}\right) :=\int\nolimits_{0}^{1}%
\left[ E_{\mathbf{A}}(t,\alpha x^{(2)}+(1-\alpha )x^{(1)})\right]
(x^{(2)}-x^{(1)})\mathrm{d}\alpha \ .  \label{V bar 0bis}
\end{equation}%
This definition is important because it is the electric field inducing
fermionic currents on bonds $(x^{(1)},x^{(2)})$ of nearest neighbors.

To simplify notation, we consider without loss of generality (w.l.o.g.)
spinless fermions with negative charge. The cases of particles with spin and
positively charged particles can be treated by exactly the same methods.
Thus, using the (minimal) coupling of $\mathbf{A}\in \mathbf{C}_{0}^{\infty
} $ to the discrete Laplacian, the discrete magnetic Laplacian is (up to a
minus sign) the self--adjoint operator
\begin{equation*}
\Delta _{\omega ,\vartheta }^{(\mathbf{A})}\equiv \Delta _{\omega ,\vartheta
}^{(\mathbf{A}(t,\cdot ))}\in \mathcal{B}(\ell ^{2}(\mathfrak{L}))\ ,\qquad
t\in \mathbb{R}\ ,
\end{equation*}%
defined\footnote{%
Observe that the sign of the coupling between the electromagnetic potential $%
\mathbf{A}\in \mathbf{C}_{0}^{\infty }$ and the laplacian is wrong in \cite[%
Eq. (2.8)]{OhmI}.} by%
\begin{equation}
\langle \mathfrak{e}_{x},\Delta _{\omega ,\vartheta }^{(\mathbf{A})}%
\mathfrak{e}_{y}\rangle =\exp \left( i\int\nolimits_{0}^{1}\left[ \mathbf{A}%
(t,\alpha y+(1-\alpha )x)\right] (y-x)\mathrm{d}\alpha \right) \langle
\mathfrak{e}_{x},\Delta _{\omega ,\vartheta }\mathfrak{e}_{y}\rangle
\label{eq discrete lapla A}
\end{equation}%
for all $t\in \mathbb{R}$ and $x,y\in \mathfrak{L}$. Here, $\langle \cdot
,\cdot \rangle $ is the scalar product in $\ell ^{2}(\mathfrak{L})$ and $%
\left\{ \mathfrak{e}_{x}\right\} _{x\in \mathfrak{L}}$ is the canonical
orthonormal basis $\mathfrak{e}_{x}(y)\equiv \delta _{x,y}$ of $\ell ^{2}(%
\mathfrak{L})$. In (\ref{eq discrete lapla A}), similar to (\ref{V bar 0bis}%
), $\alpha y+(1-\alpha )x$ and $y-x$ are seen as vectors in ${\mathbb{R}}%
^{d} $. In presence of an electromagnetic field associated to an arbitrary
vector potential $\mathbf{A}\in \mathbf{C}_{0}^{\infty }$, the one--particle
Hamiltonian $(\Delta _{\omega ,\vartheta }+\lambda V_{\omega })$ at fixed $%
\omega \in \Omega $ and $\lambda ,\vartheta \in \mathbb{R}_{0}^{+}$ is
replaced with (the time--dependent one)%
\begin{equation*}
\Delta _{\omega ,\vartheta }^{(\mathbf{A})}+\lambda V_{\omega }\equiv \Delta
_{\omega ,\vartheta }^{(\mathbf{A}(t,\cdot ))}+\lambda V_{\omega }\ ,\qquad
t\in \mathbb{R}\ .
\end{equation*}

\subsection{CAR $C^{\ast }$--Algebra of the Infinite Lattice\label{section
CAR}}

Let $\mathcal{P}_{f}(\mathfrak{L})\subset 2^{\mathfrak{L}}$ be the set of
all \emph{finite} subsets of the $d$--dimensional cubic lattice $\mathfrak{L}
$. For any $\Lambda \in \mathcal{P}_{f}(\mathfrak{L})$, $\mathcal{U}%
_{\Lambda }$ is the finite dimensional $C^{\ast }$--algebra generated by $%
\mathbf{1}$ and generators $\{a_{x,\mathrm{s}}\}_{x\in \Lambda ,\mathrm{s}%
\in \mathrm{S}}$ satisfying the canonical anti--commutation relations, $%
\mathrm{S}$ being some finite set of spins. As just explained above, the
spin dependence of $a_{x,\mathrm{s}}\equiv a_{x}$ is irrelevant in our
proofs (up to trivial modifications) and, w.l.o.g., we only consider
spinless fermions, i.e., the case $\mathrm{S}=\{0\}$.

Let%
\begin{equation}
\Lambda _{L}:=\{(x_{1},\ldots ,x_{d})\in \mathfrak{L}\,:\,|x_{1}|,\ldots
,|x_{d}|\leq L\}\in \mathcal{P}_{f}(\mathfrak{L})  \label{eq:def lambda n}
\end{equation}%
for all $L\in \mathbb{R}^{+}$ and observe that $\{\mathcal{U}_{\Lambda
_{L}}\}_{L\in \mathbb{R}^{+}}$ is an increasing net of $C^{\ast }$%
--algebras. Hence, the set%
\begin{equation}
\mathcal{U}_{0}:=\underset{L\in \mathbb{R}^{+}}{\bigcup }\mathcal{U}%
_{\Lambda _{L}}  \label{simple}
\end{equation}%
of local elements is a normed $\ast $--algebra with $\left\Vert A\right\Vert
_{\mathcal{U}_{0}}=\left\Vert A\right\Vert _{\mathcal{U}_{\Lambda _{L}}}$for
all $A\in \mathcal{U}_{\Lambda _{L}}$ and $L\in \mathbb{R}^{+}$. The CAR $%
C^{\ast }$--algebra $\mathcal{U}$ of the infinite system with norm $\Vert
\cdot \Vert _{\mathcal{U}}$ is by definition the completion of the normed $%
\ast $--algebra $\mathcal{U}_{0}$. It is separable, by finite dimensionality
of $\mathcal{U}_{\Lambda }$ for $\Lambda \in \mathcal{P}_{f}(\mathfrak{L})$.

For any fixed $\theta \in \mathbb{R}/(2\pi \mathbb{Z)}$, the condition
\begin{equation}
\sigma _{\theta }(a_{x})=\mathrm{e}^{-i\theta }a_{x}
\label{definition of gauge}
\end{equation}%
defines a unique automorphism $\sigma _{\theta }$ of the $C^{\ast }$%
--algebra $\mathcal{U}$. A special role is played by $\sigma _{\pi }$.
Elements $B_{1},B_{2}\in \mathcal{U}$ satisfying $\sigma _{\pi
}(B_{1})=B_{1} $ and $\sigma _{\pi }(B_{2})=-B_{2}$ are respectively called
\emph{even} and \emph{odd}, while elements $B\in \mathcal{U}$ satisfying $%
\sigma _{\theta }(B)=B$ for any $\theta \in \lbrack 0,2\pi )$ are called
\emph{gauge invariant}. The set
\begin{equation}
\mathcal{U}^{+}:=\{B\in \mathcal{U}\;:\;B=\sigma _{\pi }(B)\}\subset
\mathcal{U}  \label{definition of even operators}
\end{equation}%
of all even elements is a $\ast $--algebra. By continuity of $\sigma
_{\theta }$, it follows that $\mathcal{U}^{+}$ is closed and hence a $%
C^{\ast }$--algebra.

\subsection{Banach Space of Short--Range Interactions\label{Section Banach
space interaction}}

An \emph{interaction} is a family $\Phi =\{\Phi _{\Lambda }\}_{\Lambda \in
\mathcal{P}_{f}(\mathfrak{L})}$ of even and self--adjoint local elements $%
\Phi _{\Lambda }=\Phi _{\Lambda }^{\ast }\in \mathcal{U}^{+}\cap \mathcal{U}%
_{\Lambda }$ with $\Phi _{\emptyset }=0$. Obviously, the set of all
interactions can be endowed with a real vector space structure:
\begin{equation*}
\left( \alpha _{1}\Phi +\alpha _{2}\Psi \right) _{\Lambda }:=\alpha _{1}\Phi
_{\Lambda }+\alpha _{2}\Psi _{\Lambda }\ ,\qquad \Lambda \in \mathcal{P}_{f}(%
\mathfrak{L})\ ,
\end{equation*}%
for any interactions $\Phi $, $\Psi $, and any real numbers $\alpha
_{1},\alpha _{2}$. We define Banach spaces of short--range interactions by
introducing specific norms for interactions, taking into account space decay.

To this end, as done in \cite{brupedraLR} by following \cite[Eqs.
(1.3)--(1.4)]{NOS}, we consider positive--valued and non--increasing decay
functions $\mathbf{F}:\mathbb{R}_{0}^{+}\rightarrow \mathbb{R}^{+}$
satisfying the following properties:

\begin{itemize}
\item \emph{Summability on }$\mathfrak{L}$\emph{.}
\begin{equation}
\left\Vert \mathbf{F}\right\Vert _{1,\mathfrak{L}}:=\underset{y\in \mathfrak{%
L}}{\sup }\sum_{x\in \mathfrak{L}}\mathbf{F}\left( \left\vert x-y\right\vert
\right) =\sum_{x\in \mathfrak{L}}\mathbf{F}\left( \left\vert x\right\vert
\right) <\infty \ .  \label{(3.1) NS}
\end{equation}

\item \emph{Bounded convolution constant.}
\begin{equation}
\mathbf{D}:=\underset{x,y\in \mathfrak{L}}{\sup }\sum_{z\in \mathfrak{L}}%
\frac{\mathbf{F}\left( \left\vert x-z\right\vert \right) \mathbf{F}\left(
\left\vert z-y\right\vert \right) }{\mathbf{F}\left( \left\vert
x-y\right\vert \right) }<\infty \ .  \label{(3.2) NS}
\end{equation}
\end{itemize}

A typical example of such a $\mathbf{F}$ for $\mathfrak{L}=\mathbb{Z}^{d}$, $%
d\in \mathbb{N}$, is the function%
\begin{equation}
\mathbf{F}\left( r\right) =\left( 1+r\right) ^{-(d+\epsilon )}\ ,\qquad r\in
\mathbb{R}_{0}^{+}\ ,  \label{example polynomial}
\end{equation}%
which has convolution constant $\mathbf{D}\leq 2^{d+1+\epsilon }\left\Vert
\mathbf{F}\right\Vert _{1,\mathfrak{L}}$ for $\epsilon \in \mathbb{R}^{+}$.
See \cite[Eq. (1.6)]{NOS} or \cite[Example 3.1]{S}. Note that the
exponential function $\mathbf{F}\left( r\right) =\mathrm{e}^{-\varsigma r}$,
$\varsigma \in \mathbb{R}^{+}$, satisfies (\ref{(3.1) NS}) but not (\ref%
{(3.2) NS}). Nevertheless, for every function $\mathbf{F}$ with bounded
convolution constant (\ref{(3.2) NS}) and any strictly positive parameter $%
\varsigma \in \mathbb{R}^{+}$, the function
\begin{equation*}
\mathbf{\tilde{F}}\left( r\right) =\mathrm{e}^{-\varsigma r}\mathbf{F}\left(
r\right) \ ,\qquad r\in \mathbb{R}_{0}^{+}\ ,
\end{equation*}%
clearly satisfies Condition (\ref{(3.2) NS}) with a convolution constant
that is no bigger than the one of $\mathbf{F}$. In fact, as observed in \cite%
[Section 3.1]{S}, the multiplication of such a\ function $\mathbf{F}$ with a
non--increasing weight $f:\mathbb{R}_{0}^{+}\rightarrow \mathbb{R}^{+}$
satisfying $f\left( r+s\right) \geq f\left( r\right) f\left( s\right) $
(logarithmically superadditive function) produces a new positive--valued and
non--increasing decay function without increasing the convolution constant $%
\mathbf{D}$.\ In all the paper, (\ref{(3.1) NS})--(\ref{(3.2) NS}) are
assumed to be satisfied.

The function $\mathbf{F}$ encodes the short--range property of interactions.
Indeed, an interaction $\Phi $ is said to be \emph{short--range} if
\begin{equation}
\left\Vert \Phi \right\Vert _{\mathcal{W}}:=\underset{x,y\in \mathfrak{L}}{%
\sup }\sum\limits_{\Lambda \in \mathcal{P}_{f}(\mathfrak{L}),\;\Lambda
\supset \{x,y\}}\frac{\Vert \Phi _{\Lambda }\Vert _{\mathcal{U}}}{\mathbf{F}%
\left( \left\vert x-y\right\vert \right) }<\infty \ .  \label{iteration0}
\end{equation}%
The map $\Phi \mapsto \Vert \Phi \Vert _{\mathcal{W}}$ defines a norm on
interactions and the space of short--range interactions w.r.t. to the decay
function $\mathbf{F}$ is, by definition, the real separable Banach space $%
\mathcal{W}\equiv (\mathcal{W},\Vert \cdot \Vert _{\mathcal{W}}\mathcal{)}$
of all interactions $\Phi $ with $\Vert \Phi \Vert _{\mathcal{W}}<\infty $.
It turns out that all $\Phi \in \mathcal{W}$ define, in a natural way,
infinite--volume quantum dynamics, i.e., they define $C^{\ast }$--dynamical
systems on $\mathcal{U}$. This fact is discussed in detail in \cite[Section
3.1]{brupedraLR}.

\subsection{Unperturbed Dynamics of Interacting Fermions\label{Section
dynamics}}

To any $\omega \in \Omega $, hopping and potential strengths $\vartheta \in
\mathbb{R}_{0}^{+}$, we associate a short--range interaction $\Psi ^{(\omega
,\vartheta )}\in \mathcal{W}$ defined as follows: Fix $\Psi ^{\mathrm{IP}%
}\in \mathcal{W}$. Then,
\begin{equation*}
\Psi _{\Lambda }^{(\omega ,\vartheta )}:=\langle \mathfrak{e}_{x},\Delta
_{\omega ,\vartheta }\mathfrak{e}_{y}\rangle a_{x}^{\ast }a_{y}+\left(
1-\delta _{x,y}\right) \langle \mathfrak{e}_{y},\Delta _{\omega ,\vartheta }%
\mathfrak{e}_{x}\rangle a_{y}^{\ast }a_{x}+\Psi _{\{x,y\}}^{\mathrm{IP}}\in
\mathcal{U}^{+}\cap \mathcal{U}_{\Lambda }
\end{equation*}%
whenever $\Lambda =\left\{ x,y\right\} $ for $x,y\in \mathfrak{L}$, and $%
\Psi _{\Lambda }^{(\omega ,\vartheta )}:=\Psi _{\Lambda }^{\mathrm{IP}}$
otherwise. Hence, in presence of bounded static potentials, the \emph{%
internal} energy observable $H_{L}^{(\omega ,\vartheta ,\lambda )}\in
\mathcal{U}^{+}\cap \mathcal{U}_{\Lambda }$ in the box $\Lambda _{L}$ of the
interacting inhomogeneous fermion system is defined by%
\begin{eqnarray}
H_{L}^{(\omega ,\vartheta ,\lambda )} &:=&\sum\limits_{\Lambda \subset
\Lambda _{L}}\Psi _{\Lambda }^{(\omega ,\vartheta )}+\lambda
\sum\limits_{x\in \Lambda _{L}}\omega _{1}(x)a_{x}^{\ast }a_{x}
\label{def H loc} \\
&=&\sum\limits_{x,y\in \Lambda _{L}}\langle \mathfrak{e}_{x},(\Delta
_{\omega ,\vartheta }+\lambda V_{\omega })\mathfrak{e}_{y}\rangle
a_{x}^{\ast }a_{y}+\sum\limits_{\Lambda \subset \Lambda _{L}}\Psi _{\Lambda
}^{\mathrm{IP}}\ ,  \notag
\end{eqnarray}%
for $\omega =(\omega _{1},\omega _{2})\in \Omega $, $\vartheta ,\lambda \in
\mathbb{R}_{0}^{+}$ and $L\in \mathbb{R}^{+}$. Observe that the first sum in
the right hand side (r.h.s.) of the second equality in (\ref{def H loc}) is
the second quantization of the one--particle operator $\Delta _{\omega
,\vartheta }+\lambda V_{\omega }$ restricted to the subspace $\ell
^{2}(\Lambda _{L})\subset \ell ^{2}(\mathfrak{L})$. The second sum in that
equality encodes all interaction mechanisms involving more than one
particle, in the box $\Lambda _{L}$.

The finite volume dynamics thus corresponds to the continuous group $\{\tau
_{t}^{(L)}\}_{t\in {\mathbb{R}}}$ of $\ast $--auto%
\-%
morphisms defined by
\begin{equation}
\tau _{t}^{(\omega ,\vartheta ,\lambda ,L)}(B)=\mathrm{e}^{itH_{L}^{(\omega
,\vartheta ,\lambda )}}B\mathrm{e}^{-itH_{L}^{(\omega ,\vartheta ,\lambda
)}}\ ,\qquad B\in \mathcal{U}\ ,  \label{definition fininte vol dynam}
\end{equation}%
for any $t\in {\mathbb{R}}$, $\omega =(\omega _{1},\omega _{2})\in \Omega $,
$\vartheta ,\lambda \in \mathbb{R}_{0}^{+}$ and $L\in \mathbb{R}^{+}$.

As explained in \cite{NS} for quantum spin systems via Lieb--Robinson
bounds, it is clear that the strong limit $L\rightarrow \infty $ of $\tau
_{t}^{(\omega ,\vartheta ,\lambda ,L)}$ is well--defined. Here, we apply
\cite[Theorem 3.6]{brupedraLR}, which yields the following assertions:

\begin{satz}[Infinite volume dynamics and its generator]
\label{Theorem Lieb-Robinson copy(3)}\mbox{
}\newline
Assume (\ref{(3.1) NS})--(\ref{(3.2) NS}). Let $\omega \in \Omega $ and $%
\vartheta _{0},\vartheta ,\lambda \in \mathbb{R}_{0}^{+}$. \newline
\emph{(i)} Infinite volume dynamics. The continuous groups $\{\tau
_{t}^{(\omega ,\vartheta ,\lambda ,L)}\}_{t\in {\mathbb{R}}}$, $L\in {%
\mathbb{R}}^{+}$, converge strongly to a $C_{0}$--group $\{\tau
_{t}^{(\omega ,\vartheta ,\lambda )}\}_{t\in {\mathbb{R}}}$ of $\ast $--auto%
\-%
morphisms with generator $\delta ^{(\omega ,\vartheta ,\lambda )}$, as $%
L\rightarrow \infty $.\newline
\emph{(ii)} Infinitesimal generator. $\delta ^{(\omega ,\vartheta ,\lambda
)} $ is a conservative closed symmetric derivation which is equal on its
core $\mathcal{U}_{0}$ to
\begin{equation*}
\delta ^{(\omega ,\vartheta ,\lambda )}(B)=i\sum\limits_{x,y\in \mathfrak{L}%
}\langle \mathfrak{e}_{x},(\Delta _{\omega ,\vartheta }+\lambda V_{\omega })%
\mathfrak{e}_{y}\rangle \left[ a_{x}^{\ast }a_{y},B\right]
+i\sum\limits_{\Lambda \in \mathcal{P}_{f}(\mathfrak{L})}\left[ \Psi
_{\Lambda }^{\mathrm{IP}},B\right] \ ,\quad B\in \mathcal{U}_{0}\ .
\end{equation*}%
Both infinite sums in the above equation converge absolutely.\newline
\emph{(iii)} Lieb--Robinson bounds. For any $\vartheta \in \lbrack
0,\vartheta _{0}]$, $t\in \mathbb{R}$, $B_{1}\in \mathcal{U}^{+}\cap
\mathcal{U}_{\Lambda ^{(1)}}$ and $B_{2}\in \mathcal{U}_{\Lambda ^{(2)}}$
with disjoint sets $\Lambda ^{(1)},\Lambda ^{(2)}\in \mathcal{P}_{f}(%
\mathfrak{L})$,
\begin{eqnarray*}
\left\Vert \left[ \tau _{t}^{(\omega ,\vartheta ,\lambda )}\left(
B_{1}\right) ,B_{2}\right] \right\Vert _{\mathcal{U}} &\leq &2\mathbf{D}%
^{-1}\left\Vert B_{1}\right\Vert _{\mathcal{U}}\left\Vert B_{2}\right\Vert _{%
\mathcal{U}}\left( \mathrm{e}^{2\mathbf{D}\left\vert t\right\vert
D_{\vartheta _{0}}}-1\right) \\
&&\times \sum_{x\in \Lambda ^{(1)}}\sum_{y\in \Lambda ^{(2)}}\mathbf{F}%
\left( \left\vert x-y\right\vert \right) \ ,
\end{eqnarray*}%
where%
\begin{equation*}
D_{\vartheta _{0}}:=\sup \left\{ \left\Vert \Psi ^{(\omega ,\vartheta
)}\right\Vert _{\mathcal{W}}:\omega \in \Omega ,\ \vartheta \in \lbrack
0,\vartheta _{0}]\right\} <\infty \ .
\end{equation*}
\end{satz}

If $\Psi ^{\mathrm{IP}},\vartheta =0$ in $\Psi ^{(\omega ,\vartheta )}$ then
$\tau ^{(\omega ,\vartheta ,\lambda )}$ becomes a family of Bogoliubov
automorphisms of $\mathcal{U}$, as described in \cite%
{OhmI,OhmII,OhmIII,OhmIV} for homogeneous hopping terms. Meanwhile,
density--density interactions resulting from the second quantization of
two--body interactions, like for instance
\begin{equation}
\sum\limits_{x,y}v\left( \left\vert x-y\right\vert \right) a_{y}^{\ast
}a_{y}a_{x}^{\ast }a_{x}\ ,  \label{density density interaction}
\end{equation}%
where $v\left( r\right) :\mathbb{R}_{0}^{+}\rightarrow \mathbb{R}^{+}$ is a
real--valued function such that%
\begin{equation*}
\underset{r\in \mathbb{R}_{0}^{+}}{\sup }\left\{ \frac{v\left( r\right) }{%
\mathbf{F}\left( r\right) }\right\} <\infty \ ,
\end{equation*}%
can be handled in our setting. Hence, (considering fermions with spins $%
\uparrow $ or $\downarrow $) our setting includes the celebrated Hubbard
model. The Coulomb potential is excluded from our analysis because it is not
summable, see Condition (\ref{(3.1) NS}). The function
\begin{equation}
v\left( r\right) =D\frac{\mathrm{e}^{-mr}}{1+r}\ ,\qquad r\in \mathbb{R}%
_{0}^{+},\ D,m\in \mathbb{R}^{+}\ ,  \label{yukawa lattice}
\end{equation}%
which is similar to the Yukawa potential for some mass $m\in \mathbb{R}^{+}$%
, is allowed by taking, for instance, the function $\mathbf{F}\left(
r\right) =D\mathrm{e}^{-\varsigma r}(1+r)^{-(d+1)}$ with $\varsigma \in
(0,m) $.

The potential (\ref{yukawa lattice}) is the physical example we have in
mind. Indeed, it is believed in theoretical physics that the Coulomb
potential is \emph{screened} by the positively charged ions which form the
lattice $\mathfrak{L}$. In \cite[Section 1.3.2]{dia-current} the authors
assert that one should first consider models with Yukawa potentials (\ref%
{yukawa lattice}) (or even finite range) to perform the thermodynamic limit
and next the limit $m\downarrow 0$ to recover the physical model. This
procedure is only justified a posteriori and we do not consider here the
highly non--trivial mathematical problem of Coulomb interactions within a
mixture of electrons and ions. Note additionally that, as we are using the
lattice $\mathfrak{L}:=\mathbb{Z}^{d}$ to represent space, (\ref{yukawa
lattice}) does not have a singularity at $r=0$, in contrast with the Coulomb
and Yukawa potentials in the continuous space $\mathbb{R}^{d}$.

\subsection{Dynamics Driven by Time--Dependent Electromagnetic Fields\label%
{Section dynamics copy(1)}}

When the electromagnetic field is switched on, i.e., for $t\geq t_{0}$, the
\emph{total} energy observable in a box $\Lambda _{L}$ that includes the
region where the electromagnetic field does not vanish equals%
\begin{equation*}
H_{L}^{(\omega ,\vartheta ,\lambda )}+W_{t}^{(\omega ,\vartheta ,\mathbf{A}%
)}\ ,
\end{equation*}%
where, for any $\omega \in \Omega $, $\vartheta \in \mathbb{R}_{0}^{+}$, $%
\mathbf{A}\in \mathbf{C}_{0}^{\infty }$ and $t\in \mathbb{R}$,
\begin{equation}
W_{t}^{(\omega ,\vartheta ,\mathbf{A})}:=\sum\limits_{x,y\in \mathfrak{L}%
}\langle \mathfrak{e}_{x},(\Delta _{\omega ,\vartheta }^{(\mathbf{A}%
)}-\Delta _{\omega ,\vartheta })\mathfrak{e}_{y}\rangle a_{x}^{\ast
}a_{y}\in \mathcal{U}^{+}\cap \mathcal{U}_{0}  \label{eq def W}
\end{equation}%
is the electromagnetic\emph{\ potential} energy observable. The finite
volume dynamics becomes non--autonomous in presence of electromagnetic
fields.

Indeed, for any time $t\in \mathbb{R}$, consider the conservative closed
symmetric derivation%
\begin{equation*}
\delta _{t}^{(\omega ,\vartheta ,\lambda ,\mathbf{A})}:=\delta ^{(\omega
,\vartheta ,\lambda )}+i\left[ W_{t}^{(\omega ,\vartheta ,\mathbf{A})},\
\cdot \ \right] \ ,
\end{equation*}%
where we recall that $\delta ^{(\omega ,\vartheta ,\lambda )}$ is the
generator of the one--parameter group $\tau ^{(\omega ,\vartheta ,\lambda
)}:=\{\tau _{t}^{(\omega ,\vartheta ,\lambda )}\}_{t\in {\mathbb{R}}}$ of $%
\ast $--auto%
\-%
morphisms, see Theorem \ref{Theorem Lieb-Robinson copy(3)}. Define also the
family
\begin{equation}
\{\mathfrak{U}_{t,s}\}_{s,t\in \mathbb{R}}\subset \mathrm{Dom}(\delta
^{(\omega ,\vartheta ,\lambda )})  \label{eq sup0}
\end{equation}%
of unitary elements by the absolutely summable series, for any $s,t\in {%
\mathbb{R}}$,
\begin{equation}
\mathfrak{U}_{t,s}:=\mathbf{1+}\sum\limits_{k\in {\mathbb{N}}}\left(
-i\right) ^{k}\int_{s}^{t}\mathrm{d}s_{1}\cdots \int_{s}^{s_{k-1}}\mathrm{d}%
s_{k}\tau _{s_{1}-t}^{(\omega ,\vartheta ,\lambda )}\left(
W_{s_{1}}^{(\omega ,\vartheta ,\mathbf{A})}\right) \cdots \tau
_{s_{k}-t}^{(\omega ,\vartheta ,\lambda )}\left( W_{s_{k}}^{(\omega
,\vartheta ,\mathbf{A})}\right) \ .  \label{eq sup1}
\end{equation}%
As explained in \cite[Eqs. (5.18)-(5.20)]{OhmI}, this series absolutely
converges in the Banach spaces $\mathcal{U}$ and%
\begin{equation*}
(\mathrm{Dom}(\delta ^{(\omega ,\vartheta ,\lambda )}),\Vert \cdot \Vert
_{\delta ^{(\omega ,\vartheta ,\lambda )}})\ ,
\end{equation*}%
$\Vert \cdot \Vert _{\delta ^{(\omega ,\vartheta ,\lambda )}}$ being the
graph norm of the closed operator $\delta ^{(\omega ,\vartheta ,\lambda )}$
with domain $\mathrm{Dom}(\delta ^{(\omega ,\vartheta ,\lambda )})$.

Now, since $\mathbf{A}\in \mathbf{C}_{0}^{\infty }$, the map $t\mapsto
W_{t}^{(\omega ,\vartheta ,\mathbf{A})}$ from $\mathbb{R}$ to $\mathcal{U}%
_{0}$ is smooth and \cite[Corollary 4.2]{brupedraLR}\ ensures the existence
of the infinite volume non--autonomous dynamics:

\begin{satz}[Non--autonomous dynamics]
\label{Theorem Lieb-Robinson}\mbox{
}\newline
Assume (\ref{(3.1) NS})--(\ref{(3.2) NS}). Let $\omega \in \Omega $, $%
\vartheta ,\lambda \in \mathbb{R}_{0}^{+}$ and $\mathbf{A}\in \mathbf{C}%
_{0}^{\infty }$. Then, there is a strongly continuous two--para%
\-%
meter family $\{\tau _{t,s}^{(\omega ,\vartheta ,\lambda ,\mathbf{A}%
)}\}_{s,t\in {\mathbb{R}}}$ of $\ast $--auto%
\-%
morphisms on $\mathcal{U}$ satisfying the following properties:\newline
\emph{(i)} Reverse \textquotedblleft cocycle\textquotedblright\ property.
\begin{equation*}
\forall s,r,t\in \mathbb{R}:\qquad \tau _{t,s}^{(\omega ,\vartheta ,\lambda ,%
\mathbf{A})}=\tau _{r,s}^{(\omega ,\vartheta ,\lambda ,\mathbf{A})}\tau
_{t,r}^{(\omega ,\vartheta ,\lambda ,\mathbf{A})}\ .
\end{equation*}%
\emph{(ii)} Non--autonomous evolution equation. It is the unique strongly
continuous two--para%
\-%
meter family of bounded operators on $\mathcal{U}$ satisfying, in the strong
sense on the dense domain $\mathcal{U}_{0}\subset \mathcal{U}$,%
\begin{equation*}
\forall s,t\in {\mathbb{R}}:\qquad \partial _{t}\tau _{t,s}^{(\omega
,\vartheta ,\lambda ,\mathbf{A})}=\tau _{t,s}^{(\omega ,\vartheta ,\lambda ,%
\mathbf{A})}\circ \delta _{t}^{(\omega ,\vartheta ,\lambda ,\mathbf{A})}\
,\qquad \tau _{s,s}^{(\omega ,\vartheta ,\lambda ,\mathbf{A})}=\mathbf{1}_{%
\mathcal{U}}\ .
\end{equation*}%
\emph{(iii)} Interaction picture. For any $s,t\in {\mathbb{R}}$,%
\begin{equation*}
\tau _{t,s}^{(\omega ,\vartheta ,\lambda ,\mathbf{A})}(B)=\tau
_{t-s}^{(\omega ,\vartheta ,\lambda )}\left( \mathfrak{U}_{t,s}^{\ast }B%
\mathfrak{U}_{t,s}\right) \ ,\qquad B\in \mathcal{U}\ .
\end{equation*}
\end{satz}

\begin{proof}
(i)--(ii) is a direct application of \cite[Corollary 4.2]{brupedraLR}. We
omit the details. Moreover, for any $\mathbf{A}\in \mathbf{C}_{0}^{\infty }$%
, the coefficients
\begin{equation}
\mathbf{w}_{x,y}\left( \eta ,t\right) :=\langle \mathfrak{e}_{x},\vartheta
(\Delta _{\omega }^{(\eta \mathbf{A})}-\Delta _{\omega })\mathfrak{e}%
_{y}\rangle \ ,\qquad x,y\in \mathfrak{L}\ ,  \notag
\end{equation}%
of the electromagnetic potential energy observable $W_{t}^{(\omega
,\vartheta ,\eta \mathbf{A})}$ are complex--valued functions of $(\eta
,t)\in \mathbb{R}^{2}$ that satisfy
\begin{equation}
\overline{\mathbf{w}_{x,y}\left( \eta ,t\right) }=\mathbf{w}_{y,x}\left(
\eta ,t\right) \qquad \text{and}\qquad \mathbf{w}_{x,x+z}(0,t)=0
\label{condition1}
\end{equation}%
for all $x,y,z\in \mathfrak{L}$ and $(\eta ,t)\in \mathbb{R}^{2}$. By the
mean value theorem, $\{\mathbf{w}_{x,y}\}_{x,y\in \mathfrak{L}}$ are also
uniformly bounded and Lipschitz continuous (as functions of times): There is
a constant $D\in \mathbb{R}^{+}$ such that, for all $\vartheta _{0}\in
\mathbb{R}^{+}$, $\omega \in \Omega $, $\vartheta \in \lbrack 0,\vartheta
_{0}]$, $\eta \in \mathbb{R}$, and $s,t\in \mathbb{R}$,%
\begin{equation}
\sup_{x,y\in \mathfrak{L}}\left\vert \mathbf{w}_{x,y}(\eta ,t)-\mathbf{w}%
_{x,y}(\eta ,s)\right\vert \leq D\eta \left\vert t-s\right\vert \ ,\quad
\sup_{x,y\in \mathfrak{L}}\sup_{\eta ,t\in \mathbb{R}}\left\vert \mathbf{w}%
_{x,y}(\eta ,t)\right\vert \leq D\ .  \label{condition2}
\end{equation}%
From (\ref{eq sup1}) and \cite[Eq. (145), Theorem 4.7 (ii)]{brupedraLR}, we
arrive at Assertion (iii).
\end{proof}

Again, for $\Psi ^{\mathrm{IP}},\vartheta =0$, $\{\tau _{t,s}^{(\omega
,\vartheta ,\lambda ,\mathbf{A})}\}_{t\geq s}$ is the two--parameter family
of Bogoliubov automorphisms of $\mathcal{U}$ described in \cite%
{OhmI,OhmII,OhmIII,OhmIV} for homogeneous hopping terms.

\subsection{Time--Dependent State of the System\label{Section initia states}}

Thermal equilibrium states are defined here to be completely passive states.
This definition \cite[Definitions 1.1, 1.3]{PW} is based on the 2nd law of
thermodynamics. In \cite{OhmVI}, the complete passivity of states will be
discussed with much more details. By \cite[Theorem 1.4]{PW}, such states are
$(\tau ^{(\omega ,\vartheta ,\lambda )},\beta )$--KMS states for some
inverse temperature, or time scale (cf. \cite{OhmVI}), $\beta \in \left[
0,\infty \right] $.

The case $\beta =0$ corresponds to the $\tau ^{(\omega ,\vartheta ,\lambda
)} $--invariant traces, also called chaotic states, whereas the limiting
case $\beta =\infty $ refers to ground states. For simplicity, we exclude
both extreme cases and only consider fermion systems at finite inverse
temperature $\beta \in \mathbb{R}^{+}$. Note also that,\ in some situations,
the parameter $\beta $ may not be arbitrarily chosen, as illustrated in \cite%
[Example 5.3.27.]{BratteliRobinson}. In fact, it is not even a priori clear
that thermal equilibrium states, in the above sense, exist for arbitrary
interacting fermion systems. As already mentioned above, the dynamics $\tau
^{(\omega ,\vartheta ,\lambda )}$ is such that, for all $\beta \in \mathbb{R}%
^{+}$, there is at least one $(\tau ^{(\omega ,\vartheta ,\lambda )},\beta )$%
--KMS state.

Indeed, by Theorem \ref{Theorem Lieb-Robinson copy(3)} (i), the continuous
group $\{\tau _{t}^{(\omega ,\vartheta ,\lambda ,L)}\}_{t\in {\mathbb{R}}}$
of $\ast $--auto%
\-%
morphisms of $\mathcal{U}$ defined by (\ref{definition fininte vol dynam})
converges strongly to the $C_{0}$--group $\tau ^{(\omega ,\vartheta ,\lambda
)}:=\{\tau _{t}^{(\omega ,\vartheta ,\lambda )}\}_{t\in {\mathbb{R}}}$.
Moreover, it is well--known that, at finite volume and any inverse
temperature $\beta \in \mathbb{R}^{+}$, the corresponding Gibbs state is the
unique KMS state associated with $\{\tau _{t}^{(\omega ,\vartheta ,\lambda
,L)}\}_{t\in {\mathbb{R}}}$. Hence, by \cite[Proposition 5.3.25.]%
{BratteliRobinson}, there is a $(\tau ^{(\omega ,\vartheta ,\lambda )},\beta
)$--KMS state $\varrho ^{(\beta ,\omega ,\vartheta ,\lambda )}$\ for every $%
\beta \in \mathbb{R}^{+}$. Depending on the interaction $\Psi ^{\mathrm{IP}}$
in $\Psi ^{(\omega ,\vartheta )}$, the sequence of Gibbs states can have a
priori several weak$^{\ast }$--accumulation points and all these limit
states are $(\tau ^{(\omega ,\vartheta ,\lambda )},\beta )$--KMS. As a
consequence, the so--called KMS condition (and thus the completely passivity
of states) does not \emph{uniquely} define the thermal equilibrium state of
the system in infinite volume.

It is easy to check that the set $\mathfrak{Q}^{(\beta ,\omega ,\vartheta
,\lambda )}\subset \mathcal{U}^{\ast }$ of $(\tau ^{(\omega ,\vartheta
,\lambda )},\beta )$--KMS states is a non--empty, sequentially weak$^{\ast }$%
--compact and convex set. For any $\beta \in \mathbb{R}^{+}$, $\omega \in
\Omega $ and $\vartheta ,\lambda \in \mathbb{R}_{0}^{+}$, $\varrho ^{(\beta
,\omega ,\vartheta ,\lambda )}$ is, by definition, an arbitrary element of $%
\mathfrak{Q}^{(\beta ,\omega ,\vartheta ,\lambda )}$ fixed once and for all.
It represents a thermal equilibrium state of the system before the
electromagnetic field is switched on.

Since $\mathbf{A}(t,x)=0$ for all $t\leq t_{0}$, similar to \cite%
{OhmI,OhmII,OhmIII,OhmIV}, the time evolution of the state of the system
thus equals%
\begin{equation}
\rho _{t}^{(\beta ,\omega ,\vartheta ,\lambda ,\mathbf{A})}:=\left\{
\begin{array}{lll}
\varrho ^{(\beta ,\omega ,\vartheta ,\lambda )} & , & \qquad t\leq t_{0}\ ,
\\
\varrho ^{(\beta ,\omega ,\vartheta ,\lambda )}\circ \tau
_{t,t_{0}}^{(\omega ,\vartheta ,\lambda ,\mathbf{A})} & , & \qquad t\geq
t_{0}\ ,%
\end{array}%
\right.  \label{time dependent state}
\end{equation}%
for any $\beta \in \mathbb{R}^{+}$, $\omega \in \Omega $, $\vartheta
,\lambda \in \mathbb{R}_{0}^{+}$ and $\mathbf{A}\in \mathbf{C}_{0}^{\infty }$%
. By stationarity of KMS states, the definition does not depend on the
particular choice of initial time $t_{0}$ with $\mathbf{A}(t,x)=0$ for all $%
t\leq t_{0}$. This time--dependent state is generally not quasi--free unless
$\Psi ^{\mathrm{IP}}=0$ in $\Psi ^{(\omega ,\vartheta )}$.

\section{Heat Production and Current Linear Response\label{Heat Production
and Current Linear Response}}

In this section we extend \emph{all} results of \cite{OhmI,OhmII} to fermion
systems with short--range interactions. This is possible because of
Lieb--Robinson bounds for multi--commutators, which yield \cite[Theorems
4.8--4.9]{brupedraLR}. By using these results, we can follow the arguments
of \cite{OhmI,OhmII} to prove, \emph{exactly in the same way}, all their
assertions in the interacting case. Therefore, we refrain from giving all
the detailed proofs and we shorten our discussions by referring to \cite%
{OhmI,OhmII} for details.

\subsection{Heat Production\label{Section Heat Production}}

Similar to \cite{OhmI} we analyze the effect of electromagnetic fields in
terms of heat production within the fermion system. This study is related to
Joule's law, which describes the rate at which resistance converts electric
energy into heat. Its mathematical formulation requires Araki's notion of
\emph{relative entropy} \cite{Araki1,Araki2}. See \cite[Section A.1]{OhmI}
for a concise account on the relative entropy in $C^{\ast }$--algebras.

In the case of the $C^{\ast }$--algebra $\mathcal{U}$, the quantum relative
entropy takes a simple form by using the net of finite dimensional $C^{\ast
} $--algebras $\{\mathcal{U}_{\Lambda }\}_{\Lambda \in \mathcal{P}_{f}(%
\mathfrak{L})}$ generating $\mathcal{U}$. Let $\Lambda \in \mathcal{P}_{f}(%
\mathfrak{L})$ and denote by $\mathrm{tr}$ the normalized trace on $\mathcal{%
U}_{\Lambda }$, also named the tracial state of $\mathcal{U}_{\Lambda }$.
For any state $\rho _{\Lambda }\in \mathcal{U}_{\Lambda }^{\ast }$, there is
a unique adjusted density matrix $\mathrm{d}_{\rho _{\Lambda }}\in \mathcal{U%
}$, that is, $\mathrm{d}_{\rho _{\Lambda }}\geq 0$, $\mathrm{tr}(\mathrm{d}%
_{\rho _{\Lambda }})=1$ and $\rho _{\Lambda }(A)=\mathrm{tr}(\mathrm{d}%
_{\rho _{\Lambda }}A)$ for all $A\in \mathcal{U}_{\Lambda }$. See, for
instance, \cite[Lemma 3.1 (i)]{Araki-Moriya}. We define by $\mathrm{supp}%
(\rho _{\Lambda })$ the smallest projection $\mathrm{P}\in \mathcal{U}%
_{\Lambda }$ such that $\rho _{\Lambda }(\mathrm{P})=1$. Then, the relative
entropy of a state $\rho _{1,\Lambda }\in \mathcal{U}_{\Lambda }^{\ast }$
w.r.t. $\rho _{2,\Lambda }\in \mathcal{U}_{\Lambda }^{\ast }$ is defined by%
\begin{equation*}
S_{\mathcal{U}_{\Lambda }}\left( \rho _{1,\Lambda }|\rho _{2,\Lambda
}\right) :=\left\{
\begin{array}{lll}
\rho _{1,\Lambda }\left( \ln \mathrm{d}_{\rho _{1,\Lambda }}-\ln \mathrm{d}%
_{\rho _{2,\Lambda }}\right) & , & \text{if }\mathrm{supp}\left( \rho
_{2,\Lambda }\right) \geq \mathrm{supp}\left( \rho _{1,\Lambda }\right) , \\
+\infty & , & \text{otherwise}\ ,%
\end{array}%
\right.
\end{equation*}%
under the convention $x\ln x|_{x=0}:=0$. We then define the relative entropy
of any state $\rho _{1}\in \mathcal{U}^{\ast }$ w.r.t. $\rho _{2}\in
\mathcal{U}^{\ast }$ by
\begin{equation*}
\mathrm{S}\left( \rho _{1}|\rho _{2}\right) :=\underset{L\rightarrow \infty }%
{\lim }S_{\mathcal{U}_{\Lambda _{L}}}\left( \rho _{1,\Lambda _{L}}|\rho
_{2,\Lambda _{L}}\right) =\underset{L\in \mathbb{R}^{+}}{\sup }S_{\mathcal{U}%
_{\Lambda _{L}}}\left( \rho _{1,\Lambda _{L}}|\rho _{2,\Lambda _{L}}\right)
\in \left[ 0,\infty \right]
\end{equation*}%
with $\rho _{1,\Lambda _{L}}$ and $\rho _{2,\Lambda _{L}}$ being the
restrictions to $\mathcal{U}_{\Lambda _{L}}$ of the states $\rho _{1}$ and $%
\rho _{2}$, respectively. As discussed in \cite{OhmI}, this limit exists and
equals Araki's relative entropy. In particular, it is a non--negative
(possibly infinite) quantity. With this, the \emph{heat production} is
defined as follows:

\begin{definition}[Heat production]
\label{Heat production definition}\mbox{ }\newline
For any $\beta \in \mathbb{R}^{+}$, $\omega \in \Omega $, $\vartheta
,\lambda \in \mathbb{R}_{0}^{+}$ and $\mathbf{A}\in \mathbf{C}_{0}^{\infty }$%
, $\mathbf{Q}^{(\omega ,\mathbf{A})}\equiv \mathbf{Q}^{(\beta ,\omega
,\vartheta ,\lambda ,\mathbf{A})}$ is defined as a map from $\mathbb{R}$ to $%
\overline{\mathbb{R}}$ by%
\begin{equation*}
\mathbf{Q}^{(\omega ,\mathbf{A})}\left( t\right) :=\beta ^{-1}\mathrm{S}%
(\rho _{t}^{(\beta ,\omega ,\vartheta ,\lambda ,\mathbf{A})}|\varrho
^{(\beta ,\omega ,\vartheta ,\lambda )})\in \left[ 0,\infty \right] \ .
\end{equation*}
\end{definition}

In \cite[Theorem 3.2]{OhmI}, that is, for $\Psi ^{\mathrm{IP}},\vartheta =0$
in the definition of $\Psi ^{(\omega ,\vartheta )}$, we prove the 1st law of
thermodynamics for the system under consideration, implying that the heat
production generated by electromagnetic fields is exactly the increase of
the \emph{internal} energy resulting from the modification of the (infinite
volume) state of the system. Such a result is generalized for the setting
considered here. Indeed, the arguments of \cite[Sections 5.3--5.4]{OhmI} --
which proves \cite[Theorem 3.2]{OhmI} -- still work, provided interparticle
interactions have some sufficiently fast polynomial decay. More precisely,
this means the following condition on the positive and non--increasing
function $\mathbf{F}:\mathbb{R}_{0}^{+}\rightarrow \mathbb{R}^{+}$:

\begin{itemize}
\item \emph{Polynomial decay.} There is a constant $\varsigma >2d$ and, for
all $m\in \mathbb{N}_{0}$, an absolutely summable sequence $\{\mathbf{u}%
_{n,m}\}_{n\in \mathbb{N}}\in \ell ^{1}(\mathbb{N})$ such that, for all $%
n\in \mathbb{N}$ with $n>m$,%
\begin{equation}
|\Lambda _{n}\backslash \Lambda _{n-1}|\sum_{z\in \Lambda _{m}}\max_{y\in
\Lambda _{n}\backslash \Lambda _{n-1}}\mathbf{F}\left( \left\vert
z-y\right\vert \right) \leq \frac{\mathbf{u}_{n,m}}{\left( 1+n\right)
^{\varsigma }}\text{ }.  \label{(3.3) NS}
\end{equation}
\end{itemize}

\noindent Examples of functions $\mathbf{F}:\mathbb{R}_{0}^{+}\rightarrow
\mathbb{R}^{+}$ satisfying (\ref{(3.1) NS})--(\ref{(3.2) NS}) and (\ref%
{(3.3) NS}) are obviously given by (\ref{example polynomial}) for
sufficiently large parameters $\epsilon \in \mathbb{R}^{+}$. Observe that
Condition (\ref{(3.3) NS}) is assumed in \cite[Theorem 4.8]{brupedraLR}.

To present now the result, in particular the derivation of the 1st law of
thermodynamics, similar to \cite[Section 3.2]{OhmI} we need to define the
\emph{total}, \emph{internal}, and electromagnetic\emph{\ potential}
energies: For any $\beta \in \mathbb{R}^{+}$, $\omega \in \Omega $, $%
\vartheta ,\lambda \in \mathbb{R}_{0}^{+}$, $\mathbf{A}\in \mathbf{C}%
_{0}^{\infty }$ and $t\in \mathbb{R}$, the \emph{total} energy increment due
to the interaction of the charged fermions with the electromagnetic field
equals%
\begin{eqnarray}
&&\lim_{L\rightarrow \infty }\left\{ \rho _{t}^{(\beta ,\omega ,\vartheta
,\lambda ,\mathbf{A})}(H_{L}^{(\omega ,\vartheta ,\lambda )}+W_{t}^{(\omega
,\vartheta ,\mathbf{A})})-\varrho ^{(\beta ,\omega ,\vartheta ,\lambda
)}(H_{L}^{(\omega ,\vartheta ,\lambda )})\right\}  \notag \\
&=&\mathbf{S}^{(\omega ,\mathbf{A})}\left( t\right) +\mathbf{P}^{(\omega ,%
\mathbf{A})}\left( t\right) \ ,  \label{lim_en_incr full}
\end{eqnarray}%
where $H_{L}^{(\omega ,\vartheta ,\lambda )}$ and $W_{t}^{(\omega ,\vartheta
,\mathbf{A})}$ are respectively the internal energy observable (\ref{def H
loc}) and the electromagnetic\emph{\ }potential energy observable (\ref{eq
def W}). Here, $\mathbf{S}^{(\omega ,\mathbf{A})}\equiv \mathbf{S}^{(\beta
,\omega ,\vartheta ,\lambda ,\mathbf{A})}$ is the \emph{internal} energy
increment defined by%
\begin{equation}
\mathbf{S}^{(\omega ,\mathbf{A})}\left( t\right) :=\lim_{L\rightarrow \infty
}\left\{ \rho _{t}^{(\beta ,\omega ,\vartheta ,\lambda ,\mathbf{A}%
)}(H_{L}^{(\omega ,\vartheta ,\lambda )})-\varrho ^{(\beta ,\omega
,\vartheta ,\lambda )}(H_{L}^{(\omega ,\vartheta ,\lambda )})\right\} \ ,
\label{entropic energy increment}
\end{equation}%
while the electromagnetic \emph{potential} energy (increment) $\mathbf{P}%
^{(\omega ,\mathbf{A})}\equiv \mathbf{P}^{(\beta ,\omega ,\vartheta ,\lambda
,\mathbf{A})}$ is%
\begin{equation}
\mathbf{P}^{(\omega ,\mathbf{A})}\left( t\right) :=\rho _{t}^{(\beta ,\omega
,\vartheta ,\lambda ,\mathbf{A})}(W_{t}^{(\omega ,\vartheta ,\mathbf{A}%
)})=\rho _{t}^{(\beta ,\omega ,\vartheta ,\lambda ,\mathbf{A}%
)}(W_{t}^{(\omega ,\vartheta ,\mathbf{A})})-\varrho ^{(\beta ,\omega
,\vartheta ,\lambda )}(W_{t_{0}}^{(\omega ,\vartheta ,\mathbf{A})})
\label{electro free energy}
\end{equation}%
for any $\beta \in \mathbb{R}^{+}$, $\omega \in \Omega $, $\vartheta
,\lambda \in \mathbb{R}_{0}^{+}$, $\mathbf{A}\in \mathbf{C}_{0}^{\infty }$
and $t\in \mathbb{R}$.

Then, by using Lieb--Robinson bounds for multi--com%
\-%
mutators \cite[Theorems 3.8--3.9]{brupedraLR} of order three in a decisive
way, we extend \cite[Theorems 3.2, 5.8]{OhmI} to fermion systems with
short--range interactions:

\begin{satz}[Heat, electromagnetic work and 1st law of thermodynamics]
\label{main 1 copy(1)}\mbox{
}\newline
Assume (\ref{(3.1) NS})--(\ref{(3.2) NS}) and (\ref{(3.3) NS}). Let $\beta
\in \mathbb{R}^{+}$, $\omega \in \Omega $, $\vartheta ,\lambda \in \mathbb{R}%
_{0}^{+}$ and $\mathbf{A}\in \mathbf{C}_{0}^{\infty }$.\newline
\emph{(i)} 1st law of thermodynamics: For any $B\in \mathcal{U}_{0}$ and $%
t\in \mathbb{R}$,%
\begin{equation*}
\delta ^{(\omega ,\vartheta ,\lambda )}\circ \tau _{t}^{(\omega ,\vartheta
,\lambda )}\left( B\right) =\lim_{L\rightarrow \infty }i[H_{L}^{(\omega
,\vartheta ,\lambda )},\tau _{t}^{(\omega ,\vartheta ,\lambda )}\left(
B\right) ]\in \mathcal{U}\
\end{equation*}%
and, for any $t\geq t_{0}$,%
\begin{equation*}
\mathbf{Q}^{(\omega ,\mathbf{A})}\left( t\right) =\mathbf{S}^{(\omega ,%
\mathbf{A})}\left( t\right) \in \mathbb{R}_{0}^{+}\ .
\end{equation*}%
\emph{(ii)} Total energy increment and electromagnetic work:\ For any $t\geq
t_{0}$,
\begin{equation*}
\mathbf{S}^{(\omega ,\mathbf{A})}\left( t\right) +\mathbf{P}^{(\omega ,%
\mathbf{A})}\left( t\right) =\int_{t_{0}}^{t}\rho _{s}^{(\beta ,\omega
,\vartheta ,\lambda ,\mathbf{A})}\left( \partial _{s}W_{s}^{(\omega
,\vartheta ,\mathbf{A})}\right) \mathrm{d}s\ .
\end{equation*}%
In particular, the maps $\mathbf{Q}^{(\omega ,\mathbf{A})}$ and $\mathbf{S}%
^{(\omega ,\mathbf{A})}$ respectively defined by Definition \ref{Heat
production definition} and (\ref{entropic energy increment}) take always
positive and finite values for all times.
\end{satz}

\begin{proof}
To simplify, we fix $\beta \in \mathbb{R}^{+}$, $\omega \in \Omega $, $%
\vartheta ,\lambda \in \mathbb{R}_{0}^{+}$, $\mathbf{A}\in \mathbf{C}%
_{0}^{\infty }$ and use the notation $\delta \equiv \delta ^{(\omega
,\vartheta ,\lambda )}$, $\tau _{t}\equiv \tau _{t}^{(\omega ,\vartheta
,\lambda )}$, $\tau _{t}^{(L)}\equiv \tau _{t}^{(\omega ,\vartheta ,\lambda
,L)}$, and
\begin{equation}
\delta ^{(L)}\left( B\right) :=i[H_{L}^{(\omega ,\vartheta ,\lambda )},B]\
,\qquad B\in \mathcal{U}\ .  \label{delta L}
\end{equation}%
See (\ref{definition fininte vol dynam}) and Theorem \ref{Theorem
Lieb-Robinson copy(3)}.\smallskip

\noindent (i) By using the stationarity of the KMS state $\varrho ^{(\beta
,\omega ,\vartheta ,\lambda )}$ w.r.t. the unperturbed dynamics as well as
Theorem \ref{Theorem Lieb-Robinson} (iii), we infer from (\ref{time
dependent state}) that, for any $t\geq t_{0}$,
\begin{equation}
\rho _{t}^{(\beta ,\omega ,\vartheta ,\lambda ,\mathbf{A})}\left( B\right)
=\varrho ^{(\beta ,\omega ,\vartheta ,\lambda )}\left( \mathfrak{U}%
_{t,t_{0}}^{\ast }B\mathfrak{U}_{t,t_{0}}\right) \ ,\qquad B\in \mathcal{U}\
.  \label{interaction picture}
\end{equation}%
Recall that $\{\mathfrak{U}_{t,t_{0}}\}_{t\geq t_{0}}\subset \mathrm{Dom}%
(\delta )$, see (\ref{eq sup0})--(\ref{eq sup1}). Assume that
\begin{equation}
\delta \left( \mathfrak{U}_{t,t_{0}}\right) =\lim_{L\rightarrow \infty
}\delta ^{(L)}\left( \mathfrak{U}_{t,t_{0}}\right) \in \mathcal{U}\ .
\label{eq pas trivial}
\end{equation}%
Then, using this together with (\ref{interaction picture}) and the
continuity of states, one gets%
\begin{equation}
\mathbf{S}^{(\omega ,\mathbf{A})}\left( t\right) =-i\varrho ^{(\beta ,\omega
,\vartheta ,\lambda )}\left( \mathfrak{U}_{t,t_{0}}^{\ast }\delta \left(
\mathfrak{U}_{t,t_{0}}\right) \right) \in \mathbb{R}\ ,
\label{Theo int pictEq}
\end{equation}%
similar to \cite[Theorem 5.5]{OhmI}. Since, by definition, $\varrho ^{(\beta
,\omega ,\vartheta ,\lambda )}$ is a $(\tau ^{(\omega ,\vartheta ,\lambda
)},\beta )$--KMS state, \cite[Theorem 1.1]{JaksicPillet} implies from (\ref%
{Theo int pictEq}) that, for any $t\geq t_{0}$,
\begin{equation*}
\mathbf{S}^{(\omega ,\mathbf{A})}\left( t\right) =\mathbf{Q}^{(\omega ,%
\mathbf{A})}\left( t\right) \geq 0\ .
\end{equation*}
In other words, we obtain the 1st law of thermodynamics in the same way one
gets \cite[Theorems 5.3, 5.5, Corollaries 5.6--5.7]{OhmI}, \emph{provided (%
\ref{eq pas trivial}) holds true.}

Equation (\ref{eq pas trivial}) is not trivial at all in the general case.
Indeed, as one can see from (\ref{eq sup1}), $\mathfrak{U}_{t,t_{0}}\in
\mathrm{Dom}(\delta )$ is generally not in $\mathcal{U}_{0}$ and, by Theorem %
\ref{Theorem Lieb-Robinson copy(3)} (ii), we only know so far that
\begin{equation}
\delta \left( B\right) =\lim_{L\rightarrow \infty }\delta ^{(L)}\left(
B\right) \in \mathcal{U}\ ,\qquad B\in \mathcal{U}_{0}\ ,  \label{eq sup2}
\end{equation}%
$\mathcal{U}_{0}$ being a core for $\delta $. By using (\ref{(3.3) NS}),
Equation (\ref{eq sup2}) can however be extended to all elements of $\tau
_{t}(\mathcal{U}_{0})$ for any $t\in \mathbb{R}$, as explained below.

Indeed, using in the autonomous case very similar arguments to \cite[Eqs.
(120)--(128)]{brupedraLR} we can apply Lieb--Robinson bounds for multi--com%
\-%
mutators of order three, i.e., \cite[Theorems 3.8--3.9]{brupedraLR},
together with the proof of \cite[Lemma 3.2]{brupedraLR} (in particular \cite[%
Eqs. (23)--(27)]{brupedraLR}) to deduce that $\{\delta \circ \tau
_{t}^{(L)}\left( B\right) \}_{L\in \mathbb{R}^{+}}$ is a Cauchy net within
the complete space $\mathcal{U}$ for any $B\in \mathcal{U}_{0}$. These
arguments are rather long to write down and there is no reason to reproduce
them here again in detail. Note only that they can be applied because of
Condition (\ref{(3.3) NS}) with $\varsigma >2d$. These inequalities are used
to get the estimate \cite[Eq. (127)]{brupedraLR} in the non--autonomous case.

By Theorem \ref{Theorem Lieb-Robinson copy(3)} (i), $\{\tau
_{t}^{(L)}\}_{L\in \mathbb{R}^{+}}$ converges strongly to $\tau _{t}$ for
every $t\in \mathbb{R}$, while Theorem \ref{Theorem Lieb-Robinson copy(3)}
(ii) says that $\delta $ is a closed operator. Therefore, for any $B\in
\mathcal{U}_{0}$ and $t\in \mathbb{R}$, $\tau _{t}\left( B\right) \in
\mathrm{Dom}(\delta )$ and the family $\{\delta \circ \tau _{t}^{(L)}\left(
B\right) \}_{L\in \mathbb{R}^{+}}$ converges to $\delta \circ \tau
_{t}\left( B\right) $, i.e., by (\ref{eq sup2}),%
\begin{equation}
\delta \circ \tau _{t}\left( B\right) =\underset{L_{2}\rightarrow \infty }{%
\lim }\ \underset{L_{1}\rightarrow \infty }{\lim }\delta ^{(L_{1})}\circ
\tau _{t}^{(L_{2})}\left( B\right) \ .  \label{eq sup6}
\end{equation}

Now, we combine Theorem\ \ref{Theorem Lieb-Robinson copy(3)} and the
Lieb--Robinson bounds \cite[Theorem 3.1]{brupedraLR} for the finite volume
dynamics with (\ref{(3.1) NS})--(\ref{(3.2) NS}), (\ref{eq sup2}) and
Lebesgue's dominated convergence theorem to compute from (\ref{eq sup6}) that%
\begin{equation*}
\delta \circ \tau _{t}\left( B\right) =i\sum\limits_{x,y\in \mathfrak{L}%
}\langle \mathfrak{e}_{x},(\Delta _{\omega ,\vartheta }+\lambda V_{\omega })%
\mathfrak{e}_{y}\rangle \left[ a_{x}^{\ast }a_{y},\tau _{t}\left( B\right) %
\right] +i\sum\limits_{\Lambda \in \mathcal{P}_{f}(\mathfrak{L})}\left[ \Psi
_{\Lambda }^{\mathrm{IP}},\tau _{t}\left( B\right) \right]
\end{equation*}%
for any $B\in \mathcal{U}_{0}$ and $t\in \mathbb{R}$. For more details,
compare for instance with \cite[Eqs. (38)--(40)]{brupedraLR}. By using
Lieb--Robinson bounds for the infinite--volume dynamics (Theorem \ref%
{Theorem Lieb-Robinson copy(3)} (iii)) together with (\ref{def H loc}) and (%
\ref{delta L}), we arrive at%
\begin{equation}
\delta \circ \tau _{t}\left( B\right) =\lim_{L\rightarrow \infty }\delta
^{(L)}\circ \tau _{t}\left( B\right) \in \mathcal{U}\ ,\qquad B\in \mathcal{U%
}_{0}\ ,\ t\in \mathbb{R}\ ,  \label{eq sup8}
\end{equation}%
which is an extension of (\ref{eq sup2}) to all elements of $\tau _{t}(%
\mathcal{U}_{0})$ for any $t\in \mathbb{R}$.

Now, to prove (\ref{eq pas trivial}) from (\ref{eq sup1}) we use (\ref{eq
sup8}) together with the fact that $\delta $ is a symmetric derivation.
Indeed, (\ref{eq sup1}) together with the equation
\begin{equation*}
\delta (B_{1}B_{2})=\delta (B_{1})B_{2}+B_{1}\delta (B_{2})\ ,\qquad
B_{1},B_{2}\in \mathrm{Dom}(\delta )\ ,
\end{equation*}%
yields
\begin{eqnarray}
\delta (\mathfrak{U}_{t,t_{0}}) &=&-\sum\limits_{k\in {\mathbb{N}}}\left(
-i\right) ^{k+1}\int_{s}^{t}\mathrm{d}s_{1}\cdots \int_{s}^{s_{k-1}}\mathrm{d%
}s_{k}\sum\limits_{j=1}^{k}\tau _{s_{1}-t}\left( W_{s_{1}}^{(\omega
,\vartheta ,\mathbf{A})}\right)  \label{eq above} \\
&&\qquad \cdots \delta \left( \tau _{s_{j}-t}\left( W_{s_{j}}^{(\omega
,\vartheta ,\mathbf{A})}\right) \right) \cdots \tau _{s_{k}-t}\left(
W_{s_{k}}^{(\omega ,\vartheta ,\mathbf{A})}\right)  \notag
\end{eqnarray}%
for any $t\geq t_{0}$. Observe that a few simple estimates on
\begin{equation*}
\delta \left( \tau _{s_{j}-t}\left( W_{s_{j}}^{(\omega ,\vartheta ,\mathbf{A}%
)}\right) \right) =\tau _{s_{j}-t}\left( \delta (W_{s_{j}}^{(\omega
,\vartheta ,\mathbf{A})})\right)
\end{equation*}%
are also needed to obtain (\ref{eq above}). We omit the details. Combined
with (\ref{(3.1) NS})--(\ref{(3.2) NS}), (\ref{eq sup8}), Theorem \ref%
{Theorem Lieb-Robinson copy(3)} (iii) and Lebesgue's dominated convergence
theorem, Equation (\ref{eq above}) implies (\ref{eq pas trivial}). Note that
$\mathbf{A}\in \mathbf{C}_{0}^{\infty }$ and $W_{t}^{(\omega ,\vartheta ,%
\mathbf{A})}\in \mathcal{U}_{0}$ for any $t\geq t_{0}$, by (\ref{eq def W}).
The proof of Assertion (i) is thus completed. \smallskip

\noindent (ii) We omit the details since it is an extension of \cite[Lemma
5.4.27.]{BratteliRobinson} to unbounded symmetric derivations $\delta
^{(\omega ,\vartheta ,\lambda )}$ already done in the proof of \cite[Theorem
5.8]{OhmI}.
\end{proof}

\begin{bemerkung}[1st law of thermodynamics]
\label{remark 1st law of Thermodynamics}\mbox{
}\newline
If Conditions (\ref{(3.1) NS})--(\ref{(3.2) NS}) and (\ref{(3.3) NS}) are
satisfied then the 1st law of thermodynamics, corresponding to Theorem \ref%
{main 1 copy(1)} (i) in our specific case, holds true for any interaction $%
\Psi \in \mathcal{W}$ and static potential $\mathbf{V}$ at thermal
equilibrium. By static potential, we mean here a collection $\mathbf{V}:=\{%
\mathbf{V}_{\left\{ x\right\} }\}_{x\in \mathfrak{L}}$ of even and
self--adjoint elements such that $\mathbf{V}_{\left\{ x\right\} }=\mathbf{V}%
_{\left\{ x\right\} }^{\ast }\in \mathcal{U}^{+}\cap \mathcal{U}_{\left\{
x\right\} }$ for all $x\in \mathfrak{L}$. See \cite{brupedraLR} for the
general setting. This observation is a nontrivial consequence of
Lieb--Robinson bounds for multi--com%
\-%
mutators of order three, see \cite[Theorems 3.8--3.9]{brupedraLR}.
\end{bemerkung}

\cite[Theorem 3.4]{OhmI} describes the behavior of the heat production at
weak electromagnetic fields of free lattice fermions in thermal equilibrium.
In the following we extend this result to fermion systems with short--range
interactions: For any $l\in \mathbb{R}^{+}$ and $\mathbf{A}\in \mathbf{C}%
_{0}^{\infty }$, we consider the space--rescaled vector potential
\begin{equation}
\mathbf{A}_{l}(t,x):=\mathbf{A}(t,l^{-1}x)\ ,\quad t\in \mathbb{R},\ x\in
\mathbb{R}^{d}\ .  \label{rescaled vector potential}
\end{equation}%
Since Ohm's law is a linear\emph{\ }response to electric fields, we also
rescale the strength of the electromagnetic potential $\mathbf{A}_{l}$ by a
real parameter $\eta \in \mathbb{R}$ and study the behavior of the heat
production in the limit $\eta \rightarrow 0$.

Conditions (\ref{condition1})--(\ref{condition2}) and \cite[Eq. (155)]%
{brupedraLR}, that is, for all parameters $\eta ,\eta _{0}\in \mathbb{R}$,
\begin{equation}
\sup_{x,y\in \mathfrak{L}}\sup_{t\in \mathbb{R}}\left\vert \mathbf{w}%
_{x,y}(\eta ,t)-\mathbf{w}_{x,y}(\eta _{0},t)\right\vert \leq D\left\vert
\eta -\eta _{0}\right\vert \ ,  \label{condition3}
\end{equation}%
are satisfied by the perturbation $W_{t}^{(\omega ,\vartheta ,\mathbf{A})}$.
By Theorem \ref{main 1 copy(1)} (i), we can apply either \cite[Theorem 4.8]%
{brupedraLR} or \cite[Theorem 4.9]{brupedraLR}, depending on the space decay
of interparticle interaction $\Psi ^{\mathrm{IP}}$, to obtain the behavior
of the heat production $\mathbf{Q}^{(\omega ,\eta \mathbf{A}_{l})}$ w.r.t. $%
\eta ,l\in \mathbb{R}^{+}$. The interaction $\Phi $ appearing in \cite[%
Theorem 4.8]{brupedraLR} and \cite[Theorem 4.9]{brupedraLR} is in this case
the sum of the interaction $\Psi ^{\mathrm{IP}}$ with an element from $%
\mathcal{W}$ taking into account the hoppings of particles and external
static potential which are encoded via
\begin{equation*}
\omega =(\omega _{1},\omega _{2})\in \Omega :=[-1,1]^{\mathfrak{L}}\times
\mathbb{D}^{\mathfrak{b}}\text{ },
\end{equation*}%
where we recall that $\mathbb{D}:=\{z\in \mathbb{C}$ $:$ $\left\vert
z\right\vert \leq 1\}$. Note that the map
\begin{equation*}
\eta \mapsto W_{s}^{(\omega ,\vartheta ,\eta \mathbf{A}_{l})}\in \mathcal{U}
\end{equation*}%
is real analytic and the conditions of \cite[Theorem 4.9]{brupedraLR} are
satisfied if the function $\mathbf{F}$ defining the norm of the space $%
\mathcal{W}$ of interactions decays sufficiently fast at large arguments.
Theorem \ref{main 1 copy(1)} (i) combined with \cite[Theorem 4.9]{brupedraLR}
implies that $\eta \mapsto \mathbf{Q}^{(\omega ,\eta \mathbf{A}_{l})}$ is a
Gevrey function on $\mathbb{R}$ in this specific case.

With the arguments using \cite[Theorems 4.8--4.9]{brupedraLR} discussed
above, \cite[Theorem 3.4]{OhmI} can straightforwardly be extended to fermion
systems with short--range interactions. Note, moreover, that \cite[Theorems
4.8--4.9]{brupedraLR} also make possible the study of non--quadratic (resp.
non--linear) corrections to Joule's law (resp. Ohm's law). The minimal
requirement to study the linear\emph{\ }response to electric fields is some
sufficiently fast polynomial decay of interparticle interactions, that is,
Condition (\ref{(3.3) NS}).

\subsection{Charge Transport Coefficients\label{Sect Trans coeef ddef}}

Fix $\omega \in \Omega $, $\vartheta ,\lambda \in \mathbb{R}_{0}^{+}$. The
paramagnetic current observables is defined by%
\begin{equation}
I_{\mathbf{x}}^{(\omega ,\vartheta )}:=-2\mathrm{Im}\left( \langle \mathfrak{%
e}_{x^{(1)}},\Delta _{\omega ,\vartheta }\mathfrak{e}_{x^{(2)}}\rangle
a_{x^{(1)}}^{\ast }a_{x^{(2)}}\right) \ ,\qquad \mathbf{x}%
:=(x^{(1)},x^{(2)})\in \mathfrak{L}^{2}\ .  \label{current observable}
\end{equation}%
Indeed, $I_{\mathbf{x}}^{(\omega ,\vartheta )}\in \mathcal{U}_{0}$ is seen
as a current observable because, by Theorem \ref{Theorem Lieb-Robinson
copy(3)} (ii),\ it satisfies the discrete continuity equation%
\begin{equation}
\partial _{t}\left( \tau _{t}^{(\omega ,\vartheta ,\lambda )}(a_{x}^{\ast
}a_{x})\right) =\tau _{t}^{(\omega ,\vartheta ,\lambda )}\left(
-\sum\limits_{y\in \mathfrak{L},|x-y|=1}I_{(x,y)}^{(\omega ,\vartheta
)}+i\sum\limits_{\Lambda \in \mathcal{P}_{f}(\mathfrak{L})}\left[ \Psi
_{\Lambda }^{\mathrm{IP}},a_{x}^{\ast }a_{x}\right] \right)
\label{current observable derivative}
\end{equation}%
for any $x\in \mathfrak{L}$ and $t\in \mathbb{R}$.

For short--range interactions satisfying
\begin{equation}
\sum\limits_{\Lambda \in \mathcal{P}_{f}(\mathfrak{L})}\left[ \Psi _{\Lambda
}^{\mathrm{IP}},a_{x}^{\ast }a_{x}\right] =0\ ,\qquad x\in \mathfrak{L}\ ,
\label{static potential}
\end{equation}%
i.e., $\Psi ^{\mathrm{IP}}$ is invariant under local gauge transformations, $%
I_{\mathbf{x}}^{(\omega ,\vartheta )}$ is thus the observable related, in
absence of external electromagnetic potentials, to the flow of particles
from the lattice site $x^{(1)}$ to the lattice site $x^{(2)}$ or the current
from $x^{(2)}$ to $x^{(1)}$. [Positively charged particles can of course be
treated in the same way.]

We also define%
\begin{equation}
P_{\mathbf{x}}^{(\omega ,\vartheta )}:=2\mathrm{Re}\left( \langle \mathfrak{e%
}_{x^{(1)}},\Delta _{\omega ,\vartheta }\mathfrak{e}_{x^{(2)}}\rangle
a_{x^{(1)}}^{\ast }a_{x^{(2)}}\right) \ ,\qquad \mathbf{x}%
:=(x^{(1)},x^{(2)})\in \mathfrak{L}^{2}\ .  \label{R x}
\end{equation}%
For real--valued $\omega _{2}(x^{(1)},x^{(2)})$, this self--adjoint element
of $\mathcal{U}_{0}$ is proportional to the second--quantization of the
adjacency matrix of the oriented graph containing exactly the oriented bonds
$(x^{(2)},x^{(1)})$ and $(x^{(1)},x^{(2)})$.

Now, for any $\beta \in \mathbb{R}^{+}$, $\omega \in \Omega $ and $\vartheta
,\lambda \in \mathbb{R}_{0}^{+}$, we introduce two important functions
associated with the above defined observables $I_{\mathbf{x}}^{(\omega
,\vartheta )}$ and $P_{\mathbf{x}}^{(\omega ,\vartheta )}$:

\begin{itemize}
\item[(p)] The paramagnetic transport coefficient $\sigma _{\mathrm{p}%
}^{(\omega )}\equiv \sigma _{\mathrm{p}}^{(\beta ,\omega ,\vartheta ,\lambda
)}$ is defined, for any $\mathbf{x},\mathbf{y}\in \mathfrak{L}^{2}$ and $%
t\in \mathbb{R}$, by%
\begin{equation}
\sigma _{\mathrm{p}}^{(\omega )}\left( \mathbf{x},\mathbf{y},t\right)
:=\int\nolimits_{0}^{t}\varrho ^{(\beta ,\omega ,\vartheta ,\lambda )}\left(
i[I_{\mathbf{y}}^{(\omega ,\vartheta )},\tau _{s}^{(\omega ,\vartheta
,\lambda )}(I_{\mathbf{x}}^{(\omega ,\vartheta )})]\right) \mathrm{d}s\ .
\label{backwards -1bis}
\end{equation}

\item[(d)] The diamagnetic transport coefficient $\sigma _{\mathrm{d}%
}^{(\omega )}\equiv \sigma _{\mathrm{d}}^{(\beta ,\omega ,\vartheta ,\lambda
)}$ is defined by%
\begin{equation}
\sigma _{\mathrm{d}}^{(\omega )}\left( \mathbf{x}\right) :=\varrho ^{(\beta
,\omega ,\vartheta ,\lambda )}\left( P_{\mathbf{x}}^{(\omega ,\vartheta
)}\right) \ ,\qquad \mathbf{x}\in \mathfrak{L}^{2}\ .
\label{backwards -1bispara}
\end{equation}
\end{itemize}

\noindent As explained in \cite[Section 3.3]{OhmII}, $\sigma _{\mathrm{p}%
}^{(\omega )}$ can be associated with a \textquotedblleft quantum current
viscosity\textquotedblright , while $\sigma _{\mathrm{d}}^{(\omega )}$ is
related to the ballistic motion of charged particles within electric fields.

For large samples (i.e., large $l\in \mathbb{R}^{+}$), we then define the
space--averaged paramagnetic transport coefficient
\begin{equation*}
t\mapsto \Xi _{\mathrm{p},l}^{(\omega )}\left( t\right) \equiv \Xi _{\mathrm{%
p},l}^{(\beta ,\omega ,\vartheta ,\lambda )}\left( t\right) \in \mathcal{B}(%
\mathbb{R}^{d})
\end{equation*}%
w.r.t. the canonical orthonormal basis $\{e_{k}\}_{k=1}^{d}$ of the
Euclidian space $\mathbb{R}^{d}$ by%
\begin{equation}
\left\{ \Xi _{\mathrm{p},l}^{(\omega )}\left( t\right) \right\} _{k,q}:=%
\frac{1}{\left\vert \Lambda _{l}\right\vert }\underset{x,y\in \Lambda _{l}}{%
\sum }\sigma _{\mathrm{p}}^{(\omega )}\left( x+e_{k},x,y+e_{q},y,t\right)
\label{average microscopic AC--conductivity}
\end{equation}%
for any $l,\beta \in \mathbb{R}^{+}$, $\omega \in \Omega $, $\vartheta
,\lambda \in \mathbb{R}_{0}^{+}$, $k,q\in \{1,\ldots ,d\}$ and $t\in \mathbb{%
R}$. The space--averaged diamagnetic transport coefficient
\begin{equation*}
\Xi _{\mathrm{d},l}^{(\omega )}\equiv \Xi _{\mathrm{d},l}^{(\beta ,\omega
,\vartheta ,\lambda )}\in \mathcal{B}(\mathbb{R}^{d})
\end{equation*}%
corresponds to the diagonal matrix defined by%
\begin{equation}
\left\{ \Xi _{\mathrm{d},l}^{(\omega )}\right\} _{k,q}:=\frac{\delta _{k,q}}{%
\left\vert \Lambda _{l}\right\vert }\underset{x\in \Lambda _{l}}{\sum }%
\sigma _{\mathrm{d}}^{(\omega )}\left( x+e_{k},x\right) \in \left[ -2\left(
\vartheta +1\right) ,2\left( \vartheta +1\right) \right] \ .
\label{average microscopic AC--conductivity dia}
\end{equation}%
Both coefficients are typically the paramagnetic and diamagnetic (in--phase)
conductivities one experimentally measures in all space directions.

The main properties of the paramagnetic transport coefficient $\Xi _{\mathrm{%
p},l}^{(\omega )}$ are given in the next theorem. To state it, we introduce
some notation: $\mathcal{B}_{+}(\mathbb{R}^{d})\subset \mathcal{B}(\mathbb{R}%
^{d})$ stands for the set of positive linear operators on $\mathbb{R}^{d}$
(i.e., symmetric operators w.r.t. to the canonical scalar product of $%
\mathbb{R}^{d}$ with non--negative eigenvalues). For any $\mathcal{B}(%
\mathbb{R}^{d})$--valued measure $\mu $ on $\mathbb{R}$, $\Vert \mu \Vert _{%
\mathrm{op}}$ denotes the positive measure on $\mathbb{R}$ defined, for any
Borel set $\mathcal{X}$, by
\begin{equation*}
\Vert \mu \Vert _{\mathrm{op}}\left( \mathcal{X}\right) :=\sup \left\{
\underset{i\in I}{\sum }\Vert \mu \left( \mathcal{X}_{i}\right) \Vert _{%
\mathrm{op}}:\{\mathcal{X}_{i}\}_{i\in I}\text{ is a finite Borel partition
of }\mathcal{X}\right\} \ .
\end{equation*}%
We say that $\mu $ is symmetric if $\mu (\mathcal{X})=\mu (-\mathcal{X})$
for any Borel set $\mathcal{X}\subset \mathbb{R}$. Additionally, for any $%
\Theta \in \mathcal{B}(\mathbb{R}^{d})$, define its symmetric and
antisymmetric parts (w.r.t. to the canonical scalar product of $\mathbb{R}$)
respectively by%
\begin{equation}
\lbrack \Theta ]_{+}:=\frac{1}{2}\left( \Theta +\Theta ^{\mathrm{t}}\right)
\text{\qquad and\qquad }[\Theta ]_{-}:=\frac{1}{2}\left( \Theta -\Theta ^{%
\mathrm{t}}\right) \text{ }.  \label{symm-antisymm Theta}
\end{equation}%
Here, $\Theta ^{\mathrm{t}}\in \mathcal{B}(\mathbb{R}^{d})$ stands for the
transpose w.r.t. the canonical scalar product of $\mathbb{R}$ of the
operator $\Theta \in \mathcal{B}(\mathbb{R}^{d})$. With these definitions we
have the following assertion:

\begin{satz}[Microscopic paramagnetic conductivity measures]
\label{lemma sigma pos type}\mbox{
}\newline
Assume (\ref{(3.1) NS})--(\ref{(3.2) NS}). For any $l,\beta \in \mathbb{R}%
^{+}$, $\omega \in \Omega $ and $\vartheta ,\lambda \in \mathbb{R}_{0}^{+}$,
there exists a (generally non--zero) symmetric $\mathcal{B}_{+}(\mathbb{R}%
^{d})$--valued measure $\mu _{\mathrm{p},l}^{(\omega )}\equiv \mu _{\mathrm{p%
},l}^{(\beta ,\omega ,\vartheta ,\lambda )}$ on $\mathbb{R}$ such that%
\begin{equation}
\int_{\mathbb{R}}\left( 1+\left\vert \nu \right\vert \right) \Vert \mu _{%
\mathrm{p},l}^{(\omega )}\Vert _{\mathrm{op}}(\mathrm{d}\nu )<\infty
\label{unifrom bound}
\end{equation}%
and%
\begin{equation*}
\lbrack \Xi _{\mathrm{p},l}^{(\omega )}(t)]_{+}=\int_{\mathbb{R}}\left( \cos
\left( t\nu \right) -1\right) \mu _{\mathrm{p},l}^{(\omega )}(\mathrm{d}\nu
)\ ,\qquad t\in \mathbb{R}\ .
\end{equation*}
\end{satz}

\begin{proof}
The assertion is proven by using the Duhamel two--point function exactly as
in \cite[Section 5.1.2]{OhmII}, taking into account only the symmetric part
of $\Xi _{\mathrm{p},l}^{(\omega )}$ and using the property%
\begin{equation}
\Xi _{\mathrm{p},l}^{(\omega )}(-t)=[\Xi _{\mathrm{p},l}^{(\omega )}(t)]^{%
\mathrm{t}}\ ,\qquad t\in \mathbb{R}\ .  \label{Ut U moins t}
\end{equation}%
Indeed, the model considered in \cite{OhmII} is time--reversal invariant. In
this case, one has $[\Xi _{\mathrm{p},l}^{(\omega )}(t)]_{+}=\Xi _{\mathrm{p}%
,l}^{(\omega )}(t)$. The more general model studied here does not have this
symmetry for non--real hopping amplitudes and/or for interparticle
interactions which are not invariant w.r.t. to time-reversal. Another
difference as compared to systems of non--interacting fermions of \cite%
{OhmII} concerns \cite[Lemma 5.10]{OhmII} which cannot be a priori extended
to the interacting case. It follows that (\ref{unifrom bound}) may not be
uniformly bounded w.r.t. $l,\beta \in \mathbb{R}^{+}$, $\omega \in \Omega $,
$\lambda \in \mathbb{R}_{0}^{+}$.
\end{proof}

One can use the Duhamel scalar product to represent (up to a real constant)
the entries of $\Xi _{\mathrm{p},l}^{(\omega )}(t)$ w.r.t. the canonical
basis of $\mathbb{R}^{d}$ as matrix elements of some unitary one--parameter
group $\{U(t)\}_{t\in \mathbb{R}}$ on a Hilbert space. Then, one gets a
similar expression as in Theorem \ref{lemma sigma pos type} for the full
coefficient $\Xi _{\mathrm{p},l}^{(\omega )}$ in terms of a measure.
However, if $[\Xi _{\mathrm{p},l}^{(\omega )}(t)]_{-}$ does not vanish for
some time $t$, then the corresponding measure is not anymore symmetric
w.r.t. $\nu $ and it does not take values in the set of symmetric matrices,
in general. Moreover, the symmetric component $[\Xi _{\mathrm{p},l}^{(\omega
)}(t)]_{+}$ of $\Xi _{\mathrm{p},l}^{(\omega )}(t)$ corresponds to the real
part of $U(t)$ and the antisymmetric component $[\Xi _{\mathrm{p}%
,l}^{(\omega )}(t)]_{-}$ to the imaginary part of the same unitary. This is
a consequence of (\ref{Ut U moins t}).

For any $\vec{w}\in \mathbb{R}^{d}$ and $t\in \mathbb{R}$, $[\Xi _{\mathrm{p}%
,l}^{(\omega )}(t)]_{-}\vec{w}\in \mathbb{R}^{d}$ and $\vec{w}$ are
orthogonal (w.r.t. the canonical scalar product of $\mathbb{R}^{d}$). So, $%
[\Xi _{\mathrm{p},l}^{(\omega )}]_{-}\neq 0$ can be physically seen as the
presence of effective magnetic fields in the system, assuming that $\Xi _{%
\mathrm{p},l}^{(\omega )}$ correctly describes the paramagnetic
conductivity. As mentioned in the proof of the above theorem, $[\Xi _{%
\mathrm{p},l}^{(\omega )}]_{-}=0$ whenever the system is invariant w.r.t.
time--reversal, but non--vanishing magnetic fields break this symmetry. The
conductivity measure describes the heat production in the presence of
electric fields which are space--homogeneous. Since currents which are
orthogonal to the applied fields are physically not expected to produce
heat, it is not astonishing that $\mu _{\mathrm{p},l}^{(\omega )}$ only
depends on the symmetric part of the conductivity $\Xi _{\mathrm{p}%
,l}^{(\omega )}$.

\begin{koro}[{Properties of $[\Xi _{\mathrm{p},l}^{(\protect\omega )}]_{+}$}]

\label{lemma sigma pos type copy(1)}\mbox{
}\newline
Assume (\ref{(3.1) NS})--(\ref{(3.2) NS}). For $l,\beta \in \mathbb{R}^{+}$,
$\omega \in \Omega $ and $\vartheta ,\lambda \in \mathbb{R}_{0}^{+}$, the
\emph{symmetric} part of $\Xi _{\mathrm{p},l}^{(\omega )}$ has the following
properties:\newline
\emph{(i)} Time--reversal symmetry of $[\Xi _{\mathrm{p},l}^{(\omega )}]_{+}$%
: $[\Xi _{\mathrm{p},l}^{(\omega )}\left( 0\right) ]_{+}=0$ and
\begin{equation*}
\lbrack \Xi _{\mathrm{p},l}^{(\omega )}(-t)]_{+}=[\Xi _{\mathrm{p}%
,l}^{(\omega )}(t)]_{+}\ ,\qquad t\in \mathbb{R}\ .
\end{equation*}%
\emph{(ii)} Negativity of $[\Xi _{\mathrm{p},l}^{(\omega )}]_{+}$:
\begin{equation*}
-[\Xi _{\mathrm{p},l}^{(\omega )}(t)]_{+}\in \mathcal{B}_{+}(\mathbb{R}%
^{d})\ ,\qquad t\in \mathbb{R}\ .
\end{equation*}%
\emph{(iii)} Ces\`{a}ro mean of $[\Xi _{\mathrm{p},l}^{(\omega )}]_{+}$:
\begin{equation*}
\underset{t\rightarrow \infty }{\lim }\ \frac{1}{t}\int_{0}^{t}[\Xi _{%
\mathrm{p},l}^{(\omega )}(s)]_{+}\mathrm{d}s=-\mu _{\mathrm{p},l}^{(\omega
)}\left( \mathbb{R}\backslash \left\{ 0\right\} \right) \ .
\end{equation*}
\end{koro}

\begin{proof}
(i)--(iii) are direct consequences of Theorem \ref{lemma sigma pos type} and
Lebesgue's dominated convergence theorem.
\end{proof}

Recall that (\ref{unifrom bound}) is a priori not uniformly bounded w.r.t.
the parameters $l,\beta \in \mathbb{R}^{+}$, $\omega \in \Omega $, $\lambda
\in \mathbb{R}_{0}^{+}$ as in \cite[Theorem 3.1]{OhmII}. In particular, the
family
\begin{equation*}
\{\Xi _{\mathrm{p},l}^{(\beta ,\omega ,\vartheta ,\lambda )}\}_{l,\beta \in
\mathbb{R}^{+},\omega \in \Omega ,\lambda \in \mathbb{R}_{0}^{+}}
\end{equation*}%
of maps from $\mathbb{R}$ to $\mathcal{B}(\mathbb{R}^{d})$ may not be
equicontinuous. Compare with \cite[Corollary 3.2 (iv)]{OhmII}. We deduce
this property for times on compacta from Lieb--Robinson bounds:

\begin{satz}[Uniform boundedness and equicontinuity properties]
\label{lemma sigma pos type copy(4)}\mbox{
}\newline
Assume (\ref{(3.1) NS})--(\ref{(3.2) NS}). Let $\vartheta _{0},\mathrm{T}\in
\mathbb{R}^{+}$. Then, for any $l,\beta \in \mathbb{R}^{+}$, $\omega \in
\Omega $, $\vartheta \in \lbrack 0,\vartheta _{0}]$, $\lambda \in \mathbb{R}%
_{0}^{+}$ and $t_{1},t_{2}\in \lbrack -\mathrm{T},\mathrm{T}]$,%
\begin{equation*}
\left\Vert \Xi _{\mathrm{p},l}^{(\omega )}\left( t_{1}\right) -\Xi _{\mathrm{%
p},l}^{(\omega )}\left( t_{2}\right) \right\Vert _{\mathrm{\max }}\leq
32\left( 1+\vartheta _{0}\right) ^{2}\left( \mathbf{D}^{-1}\left\Vert
\mathbf{F}\right\Vert _{1,\mathfrak{L}}\mathrm{e}^{4\mathbf{D}\left\vert
\mathrm{T}\right\vert D_{\vartheta _{0}}}+1\right) \left\vert
t_{2}-t_{1}\right\vert \ ,
\end{equation*}%
where $\Vert \cdot \Vert _{\mathrm{\max }}$ is the max norm of matrices.
\end{satz}

\begin{proof}
It is a direct application of Theorem \ref{Theorem Lieb-Robinson copy(3)}
(iii).
\end{proof}

As in \cite[Section 5.1.2]{OhmII} the $\mathcal{B}_{+}(\mathbb{R}^{d})$%
--valued measures $\mu _{\mathrm{p},l}^{(\omega )}$ can be represented in
terms of the spectral measure of an explicit self--adjoint operator w.r.t.
explicitly given vectors. The constant $\mu _{\mathrm{p},l}^{(\omega
)}\left( \mathbb{R}\backslash \left\{ 0\right\} \right) $ is the so--called
static admittance of linear response theory. Moreover, if $\Xi _{\mathrm{d}%
,l}^{(\omega )}$ is non--degenerated then one can construct $\mu _{\mathrm{p}%
,l}^{(\omega )}$ from the \emph{space--averaged }quantum current viscosity%
\begin{equation}
\mathbf{V}_{l}^{(\omega )}\left( t\right) \equiv \mathbf{V}_{l}^{(\beta
,\omega ,\vartheta ,\lambda )}\left( t\right) :=\left( \Xi _{\mathrm{d}%
,l}^{(\omega )}\right) ^{-1}\partial _{t}[\Xi _{\mathrm{p},l}^{(\omega
)}\left( t\right) ]_{+}\in \mathcal{B}(\mathbb{R}^{d})
\label{quantum viscosity bis bis}
\end{equation}%
for any $l,\beta \in \mathbb{R}^{+}$, $\omega \in \Omega $, $\vartheta
,\lambda \in \mathbb{R}_{0}^{+}$ and $t\in \mathbb{R}$. Indeed, $\mu _{%
\mathrm{p},l}^{(\omega )}$ is the boundary value of the (imaginary part of
the) Laplace--Fourier transform of $\Xi _{\mathrm{d},l}^{(\omega )}\mathbf{V}%
_{l}^{(\omega )}$. For more details, see \cite[Sections 3.3, 5.1.2]{OhmII}.
As compared to \cite[Sections 3.3, 5.1.2]{OhmII}, note that $\Xi _{\mathrm{p}%
,l}^{(\omega )}$ is replaced in (\ref{quantum viscosity bis bis}) with its
symmetric part.

\subsection{Microscopic Ohm's Law\label{Sect Local Ohm law}}

The diamagnetic current observable $\mathrm{I}_{\mathbf{x}}^{(\omega
,\vartheta ,\mathbf{A})}$ is defined, for any $\omega \in \Omega $, $%
\vartheta \in \mathbb{R}_{0}^{+}$, $\mathbf{A}\in \mathbf{C}_{0}^{\infty }$,
$t\geq t_{0}$ and $\mathbf{x}:=(x^{(1)},x^{(2)})\in \mathfrak{L}^{2}$, by
\begin{eqnarray}
\mathrm{I}_{\mathbf{x}}^{(\omega ,\vartheta ,\mathbf{A})}\equiv \mathrm{I}_{%
\mathbf{x}}^{(\omega ,\vartheta ,\mathbf{A}(t,\cdot ))} &:=&-2\mathrm{Im}%
\Big(\Big(\mathrm{e}^{-i\int\nolimits_{0}^{1}[\mathbf{A}(t,\alpha
x^{(2)}+(1-\alpha )x^{(1)})](x^{(2)}-x^{(1)})\mathrm{d}\alpha }-1\Big)
\notag \\
&&\qquad \qquad \qquad \times \langle \mathfrak{e}_{x^{(1)}},\Delta _{\omega
,\vartheta }\mathfrak{e}_{x^{(2)}}\rangle a_{x^{(1)}}^{\ast }a_{x^{(2)}}\Big)%
\ .  \label{current observable new}
\end{eqnarray}%
It corresponds to a correction to the current $I_{\mathbf{x}}^{(\omega
,\vartheta )}$ defined above engendered by the presence of an external
electromagnetic potential. See \cite[Section 3.1]{OhmII} for more details.

Like in \cite[Section 3]{OhmII}, w.l.o.g. we only consider
space--homogeneous (though time--dependent) electric fields in the box $%
\Lambda _{l}$ defined by (\ref{eq:def lambda n}) for $l\in \mathbb{R}^{+}$.
More precisely, let $\vec{w}:=(w_{1},\ldots ,w_{d})\in \mathbb{R}^{d}$ be
any (normalized w.r.t. the usual Euclidian norm) vector, $\mathcal{A}\in
C_{0}^{\infty }\left( \mathbb{R};\mathbb{R}\right) $ and set $\mathcal{E}%
_{t}:=-\partial _{t}\mathcal{A}_{t}$ for all $t\in \mathbb{R}$. Then, $%
\mathbf{\bar{A}}\in \mathbf{C}_{0}^{\infty }$ is defined to be the
electromagnetic potential such that the electric field equals $\mathcal{E}%
_{t}\vec{w}$ at time $t\in \mathbb{R}$ for all $x\in \left[ -1,1\right] ^{d}$
and $(0,0,\ldots ,0)$ for $t\in \mathbb{R}$ and $x\notin \left[ -1,1\right]
^{d}$. This choice yields rescaled electromagnetic potentials $\eta \mathbf{%
\bar{A}}_{l}$ as defined by (\ref{rescaled vector potential}) for $l\in
\mathbb{R}^{+}$ and $\eta \in \mathbb{R}$.

For any $l,\beta \in \mathbb{R}^{+}$, $\omega \in \Omega $, $\vartheta
,\lambda \in \mathbb{R}_{0}^{+}$, $\eta \in \mathbb{R}$, $\vec{w}\in \mathbb{%
R}^{d}$, $\mathcal{A}\in C_{0}^{\infty }\left( \mathbb{R};\mathbb{R}\right) $
and $t\geq t_{0}$, the total current density is the sum of three currents
defined from (\ref{current observable}) and (\ref{current observable new}):

\begin{itemize}
\item[(th)] The (thermal) current density
\begin{equation*}
\mathbb{J}_{\mathrm{th}}^{(\omega ,l)}\equiv \mathbb{J}_{\mathrm{th}%
}^{(\beta ,\omega ,\vartheta ,\lambda ,l)}\in \mathbb{R}^{d}
\end{equation*}%
at thermal equilibrium inside the box $\Lambda _{l}$ is defined, for any $%
k\in \{1,\ldots ,d\}$, by%
\begin{equation}
\left\{ \mathbb{J}_{\mathrm{th}}^{(\omega ,l)}\right\} _{k}:=\left\vert
\Lambda _{l}\right\vert ^{-1}\underset{x\in \Lambda _{l}}{\sum }\varrho
^{(\beta ,\omega ,\vartheta ,\lambda )}\left( I_{\left( x+e_{k},x\right)
}^{(\omega ,\vartheta )}\right) \ .  \label{free current}
\end{equation}

\item[(p)] The paramagnetic current density is the map
\begin{equation*}
t\mapsto \mathbb{J}_{\mathrm{p}}^{(\omega ,\eta \mathbf{\bar{A}}_{l})}\left(
t\right) \equiv \mathbb{J}_{\mathrm{p}}^{(\beta ,\omega ,\vartheta ,\lambda
,\eta \mathbf{\bar{A}}_{l})}\left( t\right) \in \mathbb{R}^{d}
\end{equation*}%
defined by the space average of the current increment vector inside the box $%
\Lambda _{l}$ at times $t\geq t_{0}$, that is for any $k\in \{1,\ldots ,d\}$%
,
\begin{equation}
\left\{ \mathbb{J}_{\mathrm{p}}^{(\omega ,\eta \mathbf{\bar{A}}_{l})}\left(
t\right) \right\} _{k}:=\left\vert \Lambda _{l}\right\vert ^{-1}\underset{%
x\in \Lambda _{l}}{\sum }\rho _{t}^{(\beta ,\omega ,\vartheta ,\lambda ,\eta
\mathbf{\bar{A}}_{l})}\left( I_{\left( x+e_{k},x\right) }^{(\omega
,\vartheta )}\right) -\left\{ \mathbb{J}_{\mathrm{th}}^{(\omega ,l)}\right\}
_{k}\ .  \label{finite volume current density}
\end{equation}

\item[(d)] The diamagnetic (or ballistic) current density
\begin{equation*}
t\mapsto \mathbb{J}_{\mathrm{d}}^{(\omega ,\eta \mathbf{\bar{A}}_{l})}\left(
t\right) \equiv \mathbb{J}_{\mathrm{d}}^{(\beta ,\omega ,\vartheta ,\lambda
,\eta \mathbf{\bar{A}}_{l})}\left( t\right) \in \mathbb{R}^{d}
\end{equation*}%
is defined analogously, for any $t\geq t_{0}$ and $k\in \{1,\ldots ,d\}$, by%
\begin{equation}
\left\{ \mathbb{J}_{\mathrm{d}}^{(\omega ,\eta \mathbf{\bar{A}}_{l})}\left(
t\right) \right\} _{k}:=\left\vert \Lambda _{l}\right\vert ^{-1}\underset{%
x\in \Lambda _{l}}{\sum }\rho _{t}^{(\beta ,\omega ,\vartheta ,\lambda ,\eta
\mathbf{\bar{A}}_{l})}\left( \mathrm{I}_{(x+e_{k},x)}^{(\omega ,\vartheta
,\eta \mathbf{\bar{A}}_{l})}\right) \ .
\label{finite volume current density2}
\end{equation}
\end{itemize}

\noindent For more details on the physical interpretations of these
currents, we refer to \cite[Section 3.4]{OhmII}. We will prove in \cite%
{OhmVI} that if $\omega \in \Omega $ is the realization of some ergodic
random potential and hopping amplitude then, almost surely, the thermal
current density vanishes in the limit $l\rightarrow \infty $.

Similar to \cite[Section 3.5]{OhmII}, we use the transport coefficients $\Xi
_{\mathrm{p},l}^{(\omega )}$ (\ref{average microscopic AC--conductivity})
and $\Xi _{\mathrm{d},l}^{(\omega )}$ (\ref{average microscopic
AC--conductivity dia}) to define two linear response currents%
\begin{equation*}
J_{\mathrm{p},l}^{(\omega ,\mathcal{A})}\equiv J_{\mathrm{p},l}^{(\beta
,\omega ,\vartheta ,\lambda ,\vec{w},\mathcal{A})}\qquad \text{and}\qquad J_{%
\mathrm{d},l}^{(\omega ,\mathcal{A})}\equiv J_{\mathrm{d},l}^{(\beta ,\omega
,\vartheta ,\lambda ,\vec{w},\mathcal{A})}
\end{equation*}%
with values in $\mathbb{R}^{d}$ respectively by%
\begin{eqnarray*}
J_{\mathrm{p},l}^{(\omega ,\mathcal{A})}(t) &:=&\int_{t_{0}}^{t}\left( \Xi _{\mathrm{p},l}^{(\omega )}\left( t-s\right) \vec{w}\right) \mathcal{E}_{s}\mathrm{d}s\ ,\qquad t\geq t_{0}\ , \\
J_{\mathrm{d},l}^{(\omega ,\mathcal{A})}(t) &:=&\left( \Xi _{\mathrm{d},l}^{(\omega )}\vec{w}\right) \int_{t_{0}}^{t}\mathcal{E}_{s}\mathrm{d}s\
,\qquad t\geq t_{0}\ ,
\end{eqnarray*}%
for any $l,\beta \in \mathbb{R}^{+}$, $\omega \in \Omega $, $\vartheta
,\lambda \in \mathbb{R}_{0}^{+}$, $\vec{w}\in \mathbb{R}^{d}$ and $\mathcal{A%
}\in C_{0}^{\infty }\left( \mathbb{R};\mathbb{R}\right) $. They are the
linear responses of the paramagnetic and diamagnetic current densities,
respectively:

\begin{satz}[Microscopic Ohm's law]
\label{thm Local Ohm's law}\mbox{
}\newline
Assume (\ref{(3.1) NS})--(\ref{(3.2) NS}) and (\ref{(3.3) NS}). For any $%
\vartheta _{0}\in \mathbb{R}_{0}^{+}$, $\mathcal{A}\in C_{0}^{\infty }\left(
\mathbb{R};\mathbb{R}\right) $ and $\eta \in \mathbb{R}$,
\begin{equation*}
\mathbb{J}_{\mathrm{p}}^{(\omega ,\eta \mathbf{\bar{A}}_{l})}\left( t\right)
=\eta J_{\mathrm{p},l}^{(\omega ,\mathcal{A})}(t)+\mathcal{O}\left( \eta
^{2}\right) \qquad \text{and}\qquad \mathbb{J}_{\mathrm{d}}^{(\omega ,\eta
\mathbf{\bar{A}}_{l})}\left( t\right) =\eta J_{\mathrm{d},l}^{(\omega ,%
\mathcal{A})}(t)+\mathcal{O}\left( \eta ^{2}\right) \ ,
\end{equation*}%
uniformly for $l,\beta \in \mathbb{R}^{+}$, $\omega \in \Omega $, $\vartheta
\in \lbrack 0,\vartheta _{0}]$, $\lambda \in \mathbb{R}_{0}^{+}$, $\vec{w}%
\in \mathbb{R}^{d}$ (normalized) and $t\geq t_{0}$.
\end{satz}

\begin{proof}
The assertion directly follows from \cite[Theorem 4.8]{brupedraLR}, because
of Conditions (\ref{condition1})--(\ref{condition2}) and (\ref{condition3}).
We omit the details and refer to \cite[Lemmata 5.14--5.15]{OhmII} where
similar arguments are used. See also discussions of Section \ref{Section
Heat Production}.
\end{proof}

\noindent The fact that the asymptotics obtained are uniform w.r.t. $l,\beta
\in \mathbb{R}^{+}$, $\omega \in \Omega $, $\vartheta \in \lbrack
0,\vartheta _{0}]$, $\lambda \in \mathbb{R}_{0}^{+}$ and $t\geq t_{0}$ is a
crucial technical step to get macroscopic Ohm's law when $l\rightarrow
\infty $.

Now, all the discussions of \cite[Section 3.5]{OhmII} can be reproduced for
the interacting fermion system. We refrain from doing this again. We only
mention that $\Xi _{\mathrm{p},l}^{(\omega )}$ and $\Xi _{\mathrm{d}%
,l}^{(\omega )}$ are the \emph{in--phase} paramagnetic and diamagnetic
(microscopic) conductivity of the fermion system. Indeed, one can deduce
from Theorem \ref{thm Local Ohm's law} \emph{Ohm's law} for linear response
currents as well as \emph{Green--Kubo relations}.

Note that this physical interpretation makes only sense when interparticle
interactions $\Psi ^{\mathrm{IP}}$ are invariant under local gauge
transformations, i.e., when they satisfy (\ref{static potential}). This
condition is obviously satisfied by all density--density interaction like (%
\ref{density density interaction}). Interparticle interactions related to
\textquotedblleft hoppings\textquotedblright\ of particles are not of this
type: For example, in microscopic models for superconductors one usually has
BCS--type interaction formally like
\begin{equation}
\sum_{x,y}\gamma \left( x-y\right) a_{x,\uparrow }^{\ast }a_{x,\downarrow
}^{\ast }a_{y,\downarrow }a_{y,\uparrow }\   \label{BCS term}
\end{equation}%
involving fermions of spins $\uparrow $ or $\downarrow $. However, it still
possible to handle such models in the setting we describe here, by simply
redefining the current observables in such a way that a continuity relation
like (\ref{current observable derivative}) is fulfilled. See for instance
\cite[Eq. (2.11)]{brupedrameissner} for a concrete example of this
procedure. The physical interpretation of such a new current is very natural
in that case of superconductors:\ The currents $I_{\mathbf{x}}^{(\omega
,\vartheta )}$ as defined above in (\ref{current observable}) are related to
the flow of single electrons. The correction to it due to terms like (\ref%
{BCS term}) in the interaction can be understood as being the contribution
of the flow of (Cooper--) paired electrons to the total current. We postpone
such a more general study. Note finally that, as in the case of the kinetic
term, terms like (\ref{BCS term}) should also be coupled to electromagnetic
fields, similar to Equation (\ref{eq discrete lapla A}), since they
participate to the total current.

\subsection{Microscopic Joule's Law\label{Sect Joule}}

The total energy (\ref{lim_en_incr full}) given by the electromagnetic field
to the fermion system is divided in Section \ref{Section Heat Production}
into two components $\mathbf{S}^{(\omega ,\mathbf{A})}$ (\ref{entropic
energy increment}) and $\mathbf{P}^{(\omega ,\mathbf{A})}$ (\ref{electro
free energy}) that have interesting properties in terms of heat production
and potential energy. See for instance Theorem \ref{main 1 copy(1)} and
discussions in \cite[Section 3.2]{OhmI}.

Like in \cite[Section 4.3]{OhmII} the total delivered energy can also be
naturally divided in two other components with other interesting features,
in terms of currents this time. Indeed, for any $\beta \in \mathbb{R}^{+}$, $%
\omega \in \Omega $, $\vartheta ,\lambda \in \mathbb{R}_{0}^{+}$ and $%
\mathbf{A}\in \mathbf{C}_{0}^{\infty }$, we define two further energy
increments:

\begin{itemize}
\item[(p)] The paramagnetic energy increment $\mathfrak{J}_{\mathrm{p}%
}^{(\omega ,\mathbf{A})}\equiv \mathfrak{I}_{\mathrm{p}}^{(\beta ,\omega
,\vartheta ,\lambda ,\mathbf{A})}$ is the map from $\mathbb{R}$ to $\mathbb{R%
}$ defined by%
\begin{multline*}
\mathfrak{I}_{\mathrm{p}}^{(\omega ,\mathbf{A})}\left( t\right)
:=\lim_{L\rightarrow \infty }\left\{ \rho _{t}^{(\beta ,\omega ,\vartheta
,\lambda ,\mathbf{A})}(H_{L}^{(\omega ,\vartheta ,\lambda )}+W_{t}^{(\omega
,\vartheta ,\mathbf{A})})\right. \\
\left. -\varrho ^{(\beta ,\omega ,\vartheta ,\lambda )}(H_{L}^{(\omega
,\vartheta ,\lambda )}+W_{t}^{(\omega ,\vartheta ,\mathbf{A})})\right\} \ .
\end{multline*}

\item[(d)] The diamagnetic energy (increment) $\mathfrak{I}_{\mathrm{d}%
}^{(\omega ,\mathbf{A})}\equiv \mathfrak{I}_{\mathrm{d}}^{(\beta ,\omega
,\vartheta ,\lambda ,\mathbf{A})}$ is the map from $\mathbb{R}$ to $\mathbb{R%
}$ defined by
\begin{equation*}
\mathfrak{I}_{\mathrm{d}}^{(\omega ,\mathbf{A})}\left( t\right) :=\varrho
^{(\beta ,\omega ,\vartheta ,\lambda )}(W_{t}^{(\omega ,\vartheta ,\mathbf{A}%
)})=\varrho ^{(\beta ,\omega ,\vartheta ,\lambda )}(W_{t}^{(\omega
,\vartheta ,\mathbf{A})})-\varrho ^{(\beta ,\omega ,\vartheta ,\lambda
)}(W_{t_{0}}^{(\omega ,\vartheta ,\mathbf{A})})\ .
\end{equation*}
\end{itemize}

\noindent Note that $\mathfrak{I}_{\mathrm{p}}^{(\omega ,\mathbf{A})}$
exists at all times because the total delivered energy equals%
\begin{multline*}
\lim_{L\rightarrow \infty }\left\{ \rho _{t}^{(\beta ,\omega ,\vartheta
,\lambda ,\mathbf{A})}(H_{L}^{(\omega ,\vartheta ,\lambda )}+W_{t}^{(\omega
,\vartheta ,\mathbf{A})})-\varrho ^{(\beta ,\omega ,\vartheta ,\lambda
)}(H_{L}^{(\omega ,\vartheta ,\lambda )})\right\} \\
=\mathfrak{I}_{\mathrm{p}}^{(\omega ,\mathbf{A})}\left( t\right) +\mathfrak{I%
}_{\mathrm{d}}^{(\omega ,\mathbf{A})}\left( t\right)
\end{multline*}%
for any $\beta \in \mathbb{R}^{+}$, $\omega \in \Omega $, $\vartheta
,\lambda \in \mathbb{R}_{0}^{+}$, $\mathbf{A}\in \mathbf{C}_{0}^{\infty }$
and times $t\geq t_{0}$. See (\ref{lim_en_incr full}) and Theorem \ref{main
1 copy(1)} (ii).

The term $\mathfrak{J}_{\mathrm{p}}^{(\omega ,\mathbf{A})}$ is the part of
electromagnetic work implying a change of the internal state of the system.
By contrast, the second one $\mathfrak{I}_{\mathrm{d}}^{(\omega ,\mathbf{A}%
)} $ is the electromagnetic\emph{\ }potential energy of the fermion system
in the thermal equilibrium state. These two energy increments are directly
related to paramagnetic and diamagnetic currents, respectively. For more
details, see discussions in \cite[Sections 4.3--4.4]{OhmII}.

To present this, it is convenient to define the subset
\begin{equation*}
\mathfrak{K}:=\left\{ \mathbf{x}=(x^{(1)},x^{(2)})\in \mathfrak{L}^{2}\ :\
|x^{(1)}-x^{(2)}|=1\right\}
\end{equation*}%
of \emph{oriented} bonds (cf. (\ref{n-o bonds})). Then, by Theorem \ref{thm
Local Ohm's law}, for each $l,\beta \in \mathbb{R}^{+}$, $\omega \in \Omega $%
, $\vartheta ,\lambda \in \mathbb{R}_{0}^{+}$ and any electromagnetic
potential $\mathbf{A}\in \mathbf{C}_{0}^{\infty }$, the electric field in
its integrated form $\mathbf{E}_{t}^{\eta \mathbf{A}_{l}}$ (cf. (\ref{V bar
0})--(\ref{V bar 0bis}) and (\ref{rescaled vector potential})) implies
paramagnetic and diamagnetic currents with linear coefficients being
respectively equal to%
\begin{eqnarray}
J_{\mathrm{p},l}^{(\omega ,\mathbf{A})}(t,\mathbf{x})&:=&\frac{1}{2}\int_{t_{0}}^{t}\underset{\mathbf{y}\in \mathfrak{K}}{\sum }\sigma _{\mathrm{p}}^{(\omega )}\left( \mathbf{x},\mathbf{y,}t-s\right) \mathbf{E}_{s}^{\mathbf{A}_{l}}(\mathbf{y})\mathrm{d}s\ ,  \label{current para} \\
J_{\mathrm{d},l}^{(\omega ,\mathbf{A})}(t,\mathbf{x})&:=&\int_{t_{0}}^{t}\sigma _{\mathrm{d}}^{(\omega )}\left( \mathbf{x}\right)
\mathbf{E}_{s}^{\mathbf{A}_{l}}(\mathbf{x})\mathrm{d}s\ ,
\label{current dia}
\end{eqnarray}%
at any bond $\mathbf{x}\in \mathfrak{K}$ and time $t\geq t_{0}$. Recall that
$\sigma _{\mathrm{p}}^{(\omega )}$ and $\sigma _{\mathrm{d}}^{(\omega )}$
are the microscopic charge transport coefficients defined by (\ref{backwards
-1bis})--(\ref{backwards -1bispara}). Like \cite[Theorem 4.1]{OhmII}, the
following holds:

\begin{satz}[Microscopic Joule's law]
\label{Local Ohm's law thm copy(2)}\mbox{
}\newline
Assume (\ref{(3.1) NS})--(\ref{(3.2) NS}) and (\ref{(3.3) NS}). Let $l,\beta
\in \mathbb{R}^{+}$, $\omega \in \Omega $, $\vartheta _{0},\lambda \in
\mathbb{R}_{0}^{+}$, $\vartheta \in \lbrack 0,\vartheta _{0}]$, $\eta \in
\mathbb{R}$, $\mathbf{A}\in \mathbf{C}_{0}^{\infty }$ and $t\geq t_{0}$.%
\newline
\emph{(p)} Paramagnetic energy increment:%
\begin{equation*}
\mathfrak{I}_{\mathrm{p}}^{(\omega ,\eta \mathbf{A}_{l})}\left( t\right) =%
\frac{\eta ^{2}}{2}\int\nolimits_{t_{0}}^{t}\underset{\mathbf{x}\in
\mathfrak{K}}{\sum }J_{\mathrm{p},l}^{(\omega ,\mathbf{A})}(s,\mathbf{x})%
\mathbf{E}_{s}^{\mathbf{A}_{l}}(\mathbf{x})\mathrm{d}s+\mathcal{O}(\eta
^{3}l^{d})\ .
\end{equation*}%
\emph{(d)} Diamagnetic energy:
\begin{eqnarray*}
\mathfrak{I}_{\mathrm{d}}^{(\omega ,\eta \mathbf{A}_{l})}\left( t\right) &=&-%
\frac{\eta }{2}\underset{\mathbf{x}\in \mathfrak{K}}{\sum }\varrho ^{(\beta
,\omega ,\vartheta ,\lambda )}(I_{\mathbf{x}}^{(\omega ,\vartheta )})\left(
\int\nolimits_{t_{0}}^{t}\mathbf{E}_{s}^{\mathbf{A}_{l}}(\mathbf{x})\mathrm{d%
}s\right) \\
&&+\frac{\eta ^{2}}{4}\underset{\mathbf{x}\in \mathfrak{K}}{\sum }J_{\mathrm{%
d},l}^{(\omega ,\mathbf{A})}(t,\mathbf{x})\int\nolimits_{t_{0}}^{t}\mathbf{E}%
_{s}^{\mathbf{A}_{l}}(\mathbf{x})\mathrm{d}s+\mathcal{O}(\eta ^{3}l^{d})\ .
\end{eqnarray*}%
\emph{(\textbf{Q})} Heat production -- Internal energy increment:%
\begin{eqnarray*}
\mathbf{S}^{(\omega ,\eta \mathbf{A}_{l})}\left( t\right) &=&-\frac{\eta ^{2}%
}{2}\underset{\mathbf{x}\in \mathfrak{K}}{\sum }J_{\mathrm{p},l}^{(\omega ,%
\mathbf{A})}(t,\mathbf{x})\left( \int\nolimits_{t_{0}}^{t}\mathbf{E}_{s}^{%
\mathbf{A}_{l}}(\mathbf{x})\mathrm{d}s\right) \\
&&+\mathfrak{I}_{\mathrm{p}}^{(\omega ,\eta \mathbf{A}_{l})}\left( t\right) +%
\mathcal{O}(\eta ^{3}l^{d})\ .
\end{eqnarray*}%
\emph{(\textbf{P})} Electromagnetic potential energy:
\begin{eqnarray*}
\mathbf{P}^{(\omega ,\eta \mathbf{A}_{l})}\left( t\right) &=&\frac{\eta ^{2}%
}{2}\underset{\mathbf{x}\in \mathfrak{K}}{\sum }J_{\mathrm{p},l}^{(\omega ,%
\mathbf{A})}(t,\mathbf{x})\left( \int\nolimits_{t_{0}}^{t}\mathbf{E}_{s}^{%
\mathbf{A}_{l}}(\mathbf{x})\mathrm{d}s\right) \\
&&+\mathfrak{I}_{\mathrm{d}}^{(\omega ,\eta \mathbf{A}_{l})}\left( t\right) +%
\mathcal{O}(\eta ^{3}l^{d})\ .
\end{eqnarray*}%
The correction terms of order $\mathcal{O}(\eta ^{3}l^{d})$ in assertions
(p), (d), (\textbf{Q}) and (\textbf{P}) are uniformly bounded in $\beta \in
\mathbb{R}^{+}$, $\omega \in \Omega $, $\vartheta \in \lbrack 0,\vartheta
_{0}]$, $\lambda \in \mathbb{R}_{0}^{+}$ and $t\geq t_{0}$.
\end{satz}

\begin{proof}
Up to trivial modification taking into account the support of $\mathbf{A}$,
the assertions are basically direct consequences of Theorem \ref{main 1
copy(1)} and \cite[Theorem 4.8]{brupedraLR}, because of Conditions (\ref%
{condition1})--(\ref{condition2}) and (\ref{condition3}). For more details,
we also refer to the proof of \cite[Theorem 4.1]{OhmII}.
\end{proof}

\noindent The uniformity of the above estimates w.r.t. the size $l\in
\mathbb{R}^{+}$ of the region where the external electromagnetic field is
applied is a pivotal technical step to get in a companion paper \cite{OhmVI}
Joule's law when $l\rightarrow \infty $, i.e., a macroscopic\ version of the
above result. \bigskip

\noindent \textit{Acknowledgments:} This research is supported by the agency
FAPESP under Grant 2013/13215-5 as well as by the Basque Government through
the grant IT641-13 and the BERC 2014-2017 program and by the Spanish
Ministry of Economy and Competitiveness MINECO: BCAM Severo Ochoa
accreditation SEV-2013-0323, MTM2014-53850.

\end{document}